\documentclass[runningheads,a4paper]{llncs}

\newif\ifnotanonymous 

\notanonymoustrue 

\usepackage{mathtools}
\usepackage{amsfonts} 
\usepackage[T1]{fontenc}

\usepackage[english]{babel}
\usepackage{xcolor}
\usepackage{stmaryrd}

\usepackage{tikz}
\usetikzlibrary{fit,shapes,arrows,calc,decorations.pathmorphing}
\tikzset{snake arrow/.style=
{->,
decorate,
decoration={snake,amplitude=.4mm,segment length=2mm,post length=1mm}},
}

\tikzset{
    probbox/.style={
           rectangle,
           rounded corners,
           draw=black, thick,
           inner xsep=10pt,
           minimum width=8em,
           minimum height=2em,
           text centered}
}

\usepackage[hidelinks]{hyperref}
\usepackage{cleveref}
\usepackage{comment}
\usepackage{float}

\newcommand{\mathsc}[1]{{\normalfont\textsc{#1}}}
\newcommand{\mathname}[1]{{\mathsc{#1}}}
\newcommand{\algoname}[1]{{\texttt{#1}}}
\newcommand{\Isogeny}{\ensuremath{\mathname{Isogeny}}}
\newcommand{\lIsogeny}{\ensuremath{\mathname{$\ell$-IsogenyPath}}}

\newcommand{\EndRing}{\ensuremath{\mathname{EndRing}}}
\newcommand{\MaxOrder}{\ensuremath{\mathname{MaxOrder}}}

\newcommand{\MaxEnd}{\ensuremath{\mathname{MOER}}}
\newcommand{\HomModule}{\ensuremath{\mathname{HomModule}}}
\newcommand{\OneEnd}{\ensuremath{\mathname{OneEnd}}}

\usepackage{todonotes}

\newcommand{\SSp}{\mathrm{SS}_p}
\newcommand{\dTV}{d_\mathrm{TV}}

\newcommand{\Z}{\mathbb{Z}}
\newcommand{\Q}{\mathbb{Q}}
\newcommand{\R}{\mathbb{R}}
\def\C{\mathbb{C}}
\newcommand{\F}{\mathbb{F}}

\usepackage{algorithm}
\usepackage[noend]{algpseudocode}

\DeclareMathOperator {\End}{End}
\DeclareMathOperator {\Hom}{Hom}

\DeclareMathOperator {\Nrd}{Nrd}
\DeclareMathOperator {\Trd}{Trd}
\DeclareMathOperator {\disc}{disc}

\DeclareMathOperator {\poly}{poly}

\DeclareMathOperator {\Aut}{Aut}
\DeclareMathOperator {\GL}{GL}
\DeclareMathOperator {\polylog}{polylog}

\spnewtheorem*{remark_star}{Remark}{\itshape}{\normalfont}

\begin{document}

\title{Unconditional foundations for\\ supersingular isogeny-based cryptography}
\titlerunning{Unconditional foundations for isogeny-based cryptography}

\ifnotanonymous
\author{Arthur Herlédan Le Merdy \orcidID{0009-0007-6116-6863}  \and \\  Benjamin Wesolowski \orcidID{0000-0003-1249-6077}}
\authorrunning{A. Herlédan Le Merdy and B. Wesolowski}

\institute{ENS de Lyon, CNRS, UMPA, UMR 5669, Lyon, France}
\fi

\index{Herlédan Le Merdy, Arthur}
\index{Wesolowski, Benjamin}
\maketitle

\begin{abstract}

In this paper, we prove that the supersingular isogeny problem ($\Isogeny$), endomorphism ring problem ($\EndRing$) and maximal order problem ($\MaxOrder$) are equivalent under probabilistic polynomial time reductions, unconditionally. 

Isogeny-based cryptography is founded on the presumed hardness of these problems, and their interconnection is at the heart of the design and analysis of cryptosystems like the SQIsign digital signature scheme.
Previously known reductions relied on unproven assumptions such as the generalized Riemann hypothesis.
In this work, we present unconditional reductions, and extend this network of equivalences to the problem of computing the lattice of all isogenies between two supersingular elliptic curves (\HomModule).

For cryptographic applications, one requires computational problems to be hard \emph{on average} for random instances. It is well-known that if \Isogeny\ is hard (in the worst case), then it is hard for random instances. We extend this result by proving that if any of the above-mentionned classical problems is hard in the worst case, then all of them are hard on average. In particular, if there exist hard instances of \Isogeny, then all of \Isogeny, \EndRing, \MaxOrder\ and \HomModule\ are hard on average.

\end{abstract}


\begin{section}{Introduction}

\noindent
Isogeny-based cryptography, a branch of post-quantum cryptography, rests on the presumed hardness of a few interconnected computational problems: variations around the supersingular \emph{isogeny problem} ($\Isogeny$) or the \emph{endomorphism ring problem} ($\EndRing$). A collection of ``fundamental problems'' has grown with our understanding of the field and with the needs of new cryptosystems. Some, like $\OneEnd$ are well-suited for security proofs, their hardness serving as a lower bound on the security of cryptographic schemes. Others, like $\EndRing$, are better suited for attacks, thereby serving as upper bounds on the security. And some, like $\MaxOrder$, reframe these problems in a radically different language, deepening our understanding of the field.

Connecting these problems, proving computational reductions, or even equivalences, has thus become a central and fruitful line of research.
Most notably, the equivalence between the isogeny problem and the endomorphism ring problem~\cite{EC:EHLMP18,FOCS:Wesolowski21} motivated the design of \emph{SQIsign}~\cite{AC:DKLPW20}, today the most compact post-quantum digital signature scheme.

While all these fundamental problems are considered to be equivalent, only few of the computational reductions linking them are fully, unconditionally proven. Almost all previously-known results rely on an unproven assumption: the generalized Riemann hypothesis (GRH).
In this paper, we prove that all the aforementioned problems, and more, are in fact equivalent under classical, probabilistic polynomial time reductions, unconditionally.

\begin{subsection}{Contribution}
The main results of this paper are Theorem~\ref{theo:everything-is-equivalent} and Theorem~\ref{theo:worst-case-to-average-case} below. Theorem~\ref{theo:everything-is-equivalent}, summarized in Figure~\ref{fig_generalcase}, establishes the \emph{unconditional} equivalence of fundamental problems of isogeny-based cryptography.
Theorem~\ref{theo:worst-case-to-average-case} establishes their average-case hardness: if any of them is hard in the worst case, then all of them are hard on average, for uniformly random instances.

Let us start by informally introducing the computational problems at hand. Formal definitions are provided in Section~\ref{subsec:problem definitions}. The central objects of interest are so-called \emph{supersingular elliptic curves}. \emph{Isogenies} are morphisms between elliptic curves. \emph{Endomorphisms} of an elliptic curve $E$ are isogenies from $E$ to itself; they form a ring, written $\End(E)$.
\begin{itemize}
\item \OneEnd\ (the One Endomorphism problem): Given a supersingular elliptic curve $E$, find an endomorphism that is not a scalar multiplication, i.e., an element of $\End(E) \setminus \Z$. 
\item \EndRing\ (the Endomorphism Ring problem): Given a supersingular elliptic curve $E$, compute a basis of the endomorphism ring $\End(E)$.
\item \Isogeny\ (the Isogeny problem): Given two supersingular elliptic curves $E$ and $E'$, find an isogeny from $E$ to $E'$.
\item \lIsogeny\ (the $\ell$-Isogeny Path problem): Given two supersingular elliptic curves $E$ and $E'$, and a prime number $\ell$, find an isogeny $E \to E'$ of degree a power of $\ell$ (i.e., an $\ell$-isogeny path from $E$ to $E'$).
\item \HomModule\ (the Homomorphism Module problem): Given two supersingular elliptic curves $E$ and $E'$, compute a basis of the lattice $\Hom(E, E')$ of all isogenies from $E$ to $E'$. This problem has received very little attention so far. It appears to have never been formally introduced in the literature, yet has implicitly played a role.
\item \MaxOrder\ (the Maximal Order problem): Given a supersingular elliptic curve $E$, find some ``abstract ring'' $\mathcal O$ isomorphic to $\End(E)$. More precisely, $\End(E)$ is known to be isomorphic to a \emph{maximal order} $\mathcal O$ in a quaternion algebra $B_{p,\infty}$.
The $\MaxOrder$ problem asks to find an order in $B_{p,\infty}$ isomorphic to $\End(E)$.
To resolve an ambiguity in previous literature (see \Cref{sec:maxorder}), we formalize two variants: \MaxOrder, where the solver is free to choose his own model for $B_{p,\infty}$, and $\MaxOrder_\mathcal Q$, where the solution has to be in a model specified by an algorithm $\mathcal Q$.
\item \MaxEnd\ (the Maximal Order and Endomorphism Ring problem): Given a supersingular elliptic curve $E$, compute a basis of the endomorphism ring $\End(E)$, together with an isomorphism with an order in the quaternion algebra $B_{p,\infty}$.

\end{itemize}

\begin{figure}[h]
\centering

\begin{tikzpicture}

\node [probbox] (lIsogeny) at (5,-2.5) {\lIsogeny};
\node [probbox] (Isogeny) at (1,-2.5) {\Isogeny};
\node [probbox] (EndRing) at (1,0) {\EndRing};
\node [probbox] (OneEnd) at (-3,0) {\OneEnd};
\node [probbox] (MOER) at (5,0) {\MaxEnd};
\node [probbox] (MaxOrder) at (1,2.5) {\MaxOrder};
\node [probbox] (HomModule) at (-3,-2.5) {\HomModule};
\node [probbox] (MaxOrderQ) at (5,2.5) {$\MaxOrder_\mathcal{Q}$};

\draw [->,>=latex,very thin] (Isogeny) -- (lIsogeny);
\draw [->,>=latex,very thick] (HomModule) to[bend right = 10] node[below,midway] {\Cref{prop:HomModule-to-Isogeny}} (Isogeny);
\draw [->,>=latex,very thin] (Isogeny) to[bend right = 10] (HomModule);
\draw [->,>=latex, very thin] (Isogeny) to[bend right = 20] node[right,midway] {\Cref{prop:Isogeny-to-MOER}} (MOER) ;

\draw [->,>=latex,very thin] (OneEnd) to[bend right = 20] node[midway,left]{\cite{EC:PagWes24}} (Isogeny);
\draw [->,>=latex,very thin] (OneEnd) to[bend left = 10] (EndRing) {};
\draw [->,>=latex,very thick] (EndRing) to[bend left = 10] node[midway,below]{\cite{EC:PagWes24}} (OneEnd);
\draw [<-,>=latex, very thick] (MaxOrder) to[bend right = 20] node[left,midway] {\Cref{prop:OneEndtoMaxOrder}} (OneEnd);

\draw [->,>=latex, very thin] (EndRing) to[bend left = 10] (MOER);
\draw [->,>=latex, very thin] (MOER) to[bend left = 10] node[below,midway] {\Cref{prop:MOER-to-EndRing}} (EndRing);

\draw [->,>=latex, very thin] (MaxOrder) to[bend right = 10] (MOER);
\draw [->,>=latex, very thin] (MaxOrder) to[bend left = 10] (MaxOrderQ);
\draw [->,>=latex, very thin] (MaxOrderQ) to[bend left = 10] node[below,midway] {\Cref{prop:maxorder-equiv-maxorderQ}} (MaxOrder);

\draw [->, thick, {snake arrow}] (EndRing) to (OneEnd);

\end{tikzpicture}

\caption{ Summary of the relations between fundamental isogeny-based problems.
\label{fig_generalcase}
All arrows are unconditional classical polynomial time reductions.
Thin arrows have a $O(1)$ query-complexity, and thick arrows have a $\polylog(p)$ query-complexity.
Reductions with no reference are trivial, and all others are proved in the associated reference.
Reductions involving $\MaxOrder_\mathcal Q$ require oracle access to $\mathcal Q$.
The snake arrow corresponds to a quantum reduction with $O(1)$ query-complexity; this result was proven in~\cite[Theorem 5]{C:ABDPW25}.}
\end{figure}
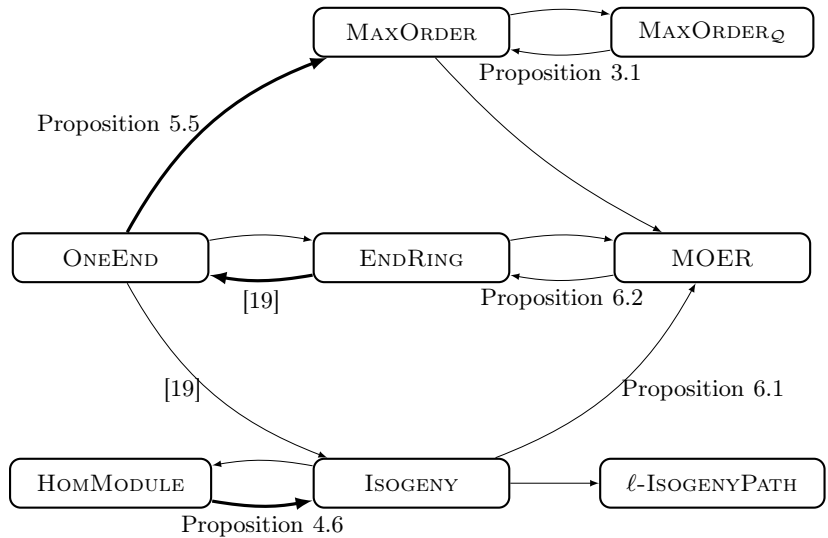

\begin{theorem}
\label{theo:everything-is-equivalent}
The problems \Isogeny, \EndRing, \MaxOrder, $\MaxOrder_\mathcal{Q}$, \HomModule, \OneEnd\ and \MaxEnd,  are all equivalent under probabilistic polynomial time reductions.
Reductions involving $\MaxOrder_\mathcal{Q}$ require oracle access to $\mathcal{Q}$.
\end{theorem}

\subsubsection*{Map of the proof of \protect{Theorem~\ref{theo:everything-is-equivalent}}}
The strategy consists in proving the computational reductions exhibited in Figure~\ref{fig_generalcase}.
Each arrow represents a computational reduction, and comes with a pointer to the proof.
Note that our unconditional reductions are substantially different from the existing conditional reductions. To eliminate any reliance on GRH, we avoid all arguments that rely on the ``good'' distribution of numbers represented by quadratic forms. This forbids us from using the powerful tools of KLPT-type algorithms~\cite{kohel_quaternion_2014}. In particular, we need to construct different ``paths'' in the network of reductions, and develop new types of arguments. We prove the reductions in the following order.
\begin{itemize}
\item The novel distinction between the two computational problems $\MaxOrder$ and $\MaxOrder_{\mathcal Q}$ is discussed in Section~\ref{sec:maxorder}. Their equivalence, proved in Proposition~\ref{prop:maxorder-equiv-maxorderQ}, hinges on a recent result \cite[Proposition 4.1]{csahok_explicit_2022} to compute isomorphisms between quaternion algebras, when some maximal orders are known in each.
\item The \HomModule\ problem is the object of Section~\ref{sec:hommod}, where we prove that it reduces to \Isogeny. Given two curves $E_1$ and $E_2$, the strategy is the following. First, we exploit the reduction from $\EndRing$ to $\Isogeny$ proved in \cite{EC:PagWes24} to compute bases of $\End(E_1)$ and $\End(E_2)$. Then, we solve \Isogeny\ again to find some $\varphi : E_1 \to E_2$. Through algebraic arguments, we prove that one can extract a basis of $\Hom(E_1,E_2)$ from the data of $\End(E_1)$, $\End(E_2)$, and $\varphi$.
\item The reduction from $\OneEnd$ to $\MaxOrder$ is the object of Section~\ref{sec:oneend to maxorder}. Navigating between a problem which deals with endomorphisms (like \OneEnd) and another which deals with purely quaternionic data (like \MaxOrder) typically requires to connect instances to some special elliptic curve $E_0$ for which both $\End(E_0)$ and its embedding in the quaternions are already known. This curve $E_0$ provides an ``endomorphism/quaternion'' dictionary. Without GRH, there is no guarantee that a special curve $E_0$ can be found. We thus need to develop a new strategy. To reduce $\OneEnd$ (say on some input $E$) to $\MaxOrder$, we solve $\MaxOrder$ on $E$ and on a few ``close neighbours'' of $E$. Doing so, we construct a ``local'' correspondence between neighbours of $E$ and quaternionic orders, and we prove that from enough such ``local'' information, we can reconstruct a full ``endomorphism/quaternion'' dictionary.
\item The reduction from $\Isogeny$ to $\MaxEnd$ is the object of Proposition~\ref{prop:Isogeny-to-MOER}. This reduction hinges on recent advances in isogeny-based cryptography facilitating the conversion of ideals in quaternionic orders into the corresponding isogenies~\cite{EPRINT:PagRob23}.
\item Finally, the reduction from $\MaxEnd$ to $\EndRing$ is the object of Proposition~\ref{prop:MOER-to-EndRing}. Of all the reductions, this one resembles the most closely an existing reduction: the reduction from $\MaxOrder$ to $\EndRing$ in~\cite{EC:EHLMP18,FOCS:Wesolowski21}. However, these former reductions required GRH to provably avoid hard factorisations. Instead, we show that no factorisation is needed if we are free to choose our own model $(\frac{a,b}{\mathbb{Q}})$ for the quaternion algebra. The parameters $a$ and $b$ are possibly hard to factor, but it does not matter: the result \cite[Proposition 4.1]{csahok_explicit_2022} allows one to convert the solution to more standard models without factoring.
\end{itemize}

\begin{remark_star}
	The different reduction paths presented in Figure~\ref{fig_generalcase} are not unique but are sufficient to establish the unconditional equivalence between the various problems.
	There are probably alternative paths improving the query complexity of certain reductions.
	In particular, we acknowledge one of our reviewers for suggesting a possible reduction from \HomModule\ to \MaxEnd\ with $O(1)$ query-complexity.
\end{remark_star}

\subsubsection*{Discussion of \protect{Theorem~\ref{theo:everything-is-equivalent}} and comparison with previous work.}
The former state of the art is summarized in Figure~\ref{figure:stateoftheart}.
We make the following observations.
\begin{itemize}
\item The \HomModule\ problem is absent from Figure~\ref{figure:stateoftheart}: its relation to other problems was never studied before.
\item The \MaxEnd\ problem is also absent, yet it is folklore that, assuming GRH, the reductions of~\cite{FOCS:Wesolowski21} also extend to \MaxEnd.
\item Reflecting previous litterature, Figure~\ref{figure:stateoftheart} makes no distinction between the problems $\MaxOrder$ and $\MaxOrder_{\mathcal Q}$. Indeed, this distinction is only useful if one refuses to believe in GRH (see Section~\ref{sec:maxorder}).
\item There remains one reduction which is only known conditionally on GRH: the reduction from $\lIsogeny$ to $\EndRing$ (or to any other problem in our list). Indeed, by definition, $\lIsogeny$ asks to find isogenies with degree of a prescribed form. The study of isogenies of prescribed degree closely relates to the study of integers represented by certain quadratic forms. GRH has a consequential impact on the distribution of integers represented by quadratic forms, and currently known unconditional results seem insufficient for the study of $\lIsogeny$.\footnote{Note that an attempt at replacing the GRH assumption with a factoring oracle is presented in~\cite{Mamah24}. A mistake in the proof has been reported, but if it can be fixed, it would link $\lIsogeny$ to the other problems under unconditional polynomial time \emph{quantum} reductions.}
\end{itemize}

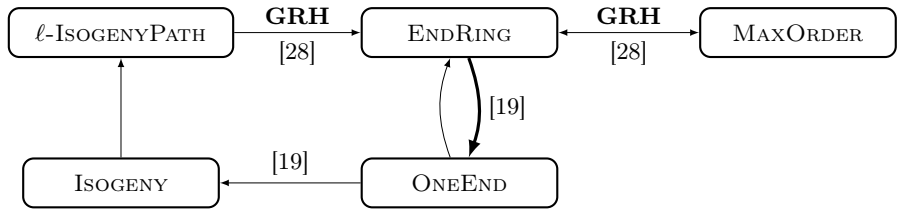
\begin{figure}[h]
	\begin{tikzpicture}
	\node [probbox] (Isogeny) at (0.5,7) {\Isogeny};
	\node [probbox] (lIsogeny) at (0.5,9) {\lIsogeny};
	\draw [->,>=latex,very thin] (Isogeny) -- (lIsogeny);
	
	\node [probbox] (EndRing) at (5,9) {\EndRing};

	\node [probbox] (MaxOrder) at (9.5,9) {\MaxOrder};
	\draw [<->,>=latex,very thin] (EndRing) -- (MaxOrder) node[above,midway, sloped]  {\textbf{GRH}} node[below,midway,sloped] {\cite{FOCS:Wesolowski21}};

	\node [probbox] (OneEnd) at (5,7) {\OneEnd};
	\draw [->,>=latex,very thin] (lIsogeny) --  (EndRing) node[above,midway, sloped] {\textbf{GRH}} node[below,midway,sloped] {\cite{FOCS:Wesolowski21}};
	\draw [->,>=latex, very thin] (OneEnd.west) -- node[midway,above,sloped]{\cite{EC:PagWes24}} (Isogeny.east);
	\draw [->,>=latex, very thin] (OneEnd) to[bend left = 20] (EndRing);
	\draw [<-,>=latex, very thick] (OneEnd) to[bend right = 20] node[midway,right]{\cite{EC:PagWes24}} (EndRing) ;

	\end{tikzpicture}
	\caption{\label{figure:stateoftheart}Former state of the art of (conditional) reductions between foundational problems of isogeny-based cryptography. All arrows are classical polynomial time reductions.
Thin arrows have a $O(1)$ query-complexity, and the thick arrow has a $\polylog(p)$ query-complexity.
Reductions with no reference are trivial, and all others are proved in the associated reference.
The \textbf{GRH} label signifies that a reduction assumes the Generalized Riemann Hypothesis.}
\end{figure}

\Cref{theo:everything-is-equivalent} implies that if hard instances exist for any one of the listed problems, then hard instances must exist for all of them. However, the security of isogeny-based schemes typically relies on the presumed hardness of \emph{random instances} of these problems.
These random instances often follow a ``natural'' distribution called the \emph{stationary distribution} (see \Cref{def:stationary} --- note that it is statistically indistinguishable from the uniform distribution). 
This distribution emerged from the use of random walks as early as the Charles--Goren--Lauter hash function~\cite{JC:ChaLauGor09}, and up to the latest advances on the SQIsign digital signature scheme~\cite{EC:DLRW24,AC:BDDLMP24}.
We are thus interested in the hardness of the \emph{average case} of the fundamental problems (see \Cref{def:averagecase}), with respect to the stationary (or uniform) distribution.

The following \Cref{theo:worst-case-to-average-case} says that there are \emph{worst-case to average-case} reductions between all of these problems. In particular, if there exists even a single hard instance for any of the listed problems, then all of the problems are hard on average --- a powerful statement for security analysis.

\begin{theorem}
\label{theo:worst-case-to-average-case}
For any pair of problems $(P,Q)$ chosen from the problems  \MaxEnd, \EndRing, \Isogeny, \OneEnd, \MaxOrder, $\MaxOrder_\mathcal{Q}$, \HomModule, and \lIsogeny\ there exists a probablistic polynomial time worst-case to average-case reduction from $P$ to $Q$.
All reductions hold unconditionally, with the two following exceptions which require the generalized Riemann hypothesis:

\begin{itemize}
	\item if $P = \lIsogeny$, or
	\item if $Q \in \{\MaxOrder, \MaxOrder_\mathcal{Q}\}$ and $p \equiv 1 \bmod 8$.
\end{itemize}
Reductions involving $\MaxOrder_\mathcal{Q}$ require oracle access to $\mathcal{Q}$.
\end{theorem}

\subsubsection*{Map of the proof of \protect{Theorem~\ref{theo:worst-case-to-average-case}}}
The strategy, carried out in Section~\ref{sec:worst-case-to-average-case-reduction}, consists in proving that the worst-case $\OneEnd$ problem reduces to the average case of any other problem in the list.
The precise network of reductions is summarized in Figure~\ref{fig_wc2ac}, page~\pageref{fig_wc2ac}.
We conclude from the fact, established in Theorem~\ref{theo:everything-is-equivalent}, that the worst case of any problem in the list reduces to a worst-case $\OneEnd$ instance (except for $P = \lIsogeny$, which relies on the conditional reduction of~\cite{FOCS:Wesolowski21}).

\subsubsection*{Discussion of \protect{Theorem~\ref{theo:worst-case-to-average-case}} and comparison with previous work.}
Some of the worst-case to average-case reductions between the problems of interest are already folklore.
It is well known, for instance, that random walks in $\ell$-isogeny graphs can be used to re-randomize an instance of the $\lIsogeny$, leading naturally to a self-reduction. This straightforward approach extends to other cases, but is not sufficient to obtain the full network of reductions proved in Theorem~\ref{theo:worst-case-to-average-case}. For instance, the reduction from the worst-case $\OneEnd$ problem to the average case $\Isogeny$, $\HomModule$ or $\lIsogeny$ problems is obtained by modifying a worst-case reduction proposed in~\cite{EC:PagWes24}. The reduction from the worst-case $\OneEnd$ problem to the average case $\MaxOrder$ or $\MaxOrder_\mathcal Q$ problems relies on recent advances facilitating the conversion between ideals and isogenies and the division of isogenies.
\end{subsection}

\ifnotanonymous
\begin{subsection}{Acknowledgements}
The authors would like to thank Travis Morrison for fruitful discussions, and for bringing our attention to \cite[Proposition 4.1]{csahok_explicit_2022}, which simplified and improved some of the results of this paper.
The authors were supported by the Agence Nationale de la Recherche under grants  ANR-22-PETQ-0008 (PQ-TLS) and ANR-22-PNCQ-0002 (HQI), and the European Research Council under grant No. 101116169 (AGATHA CRYPTY).
\end{subsection}
\fi

\end{section}

\begin{section}{Preliminaries}\label{sec:prelim}


\begin{subsection}{Notation}
We write $\Z$ for the ring of integers and $\Q$ for the field of rational numbers. For any prime power $q$, we write $\F_q$ for the finite field with $q$ elements. For any set $S$, we denote by $\#S$ its cardinality. For any field $K$, we write $\overline K$ for its algebraic closure. We write $f = O(g)$ for the classic big O notation, and use the soft O notation $\tilde O(g) = \log(g)^{O(1)} \cdot O(g)$. We also write $\poly(f_1,\dots,f_n) = (f_1 + \dots + f_n)^{O(1)}$. The function $\log$ is in base $2$. For any ring $R$, we write $R^\times$ for its group of invertible elements, and  $M_2(R)$ for the ring of $2\times 2$ matrices with coefficients in $R$.
\end{subsection}

\begin{subsection}{Quaternion algebras}
See \cite{voight_quaternion_2021} for a detailed reference on quaternion algebras.
For any $a,b\in \Q^\times$, 
the \emph{quaternion algebra} 
$B = (\frac{a,b}{\Q})$ is a ring generated by a 
$\Q$-basis $(1,i,j,k)$ satisfying the multiplication rules
$$i^2 = a, j^2 = b, k = ij = -ji.$$
In this article, we only consider quaternion algebras of this form.
For any prime $p$, we say that $B$ is \emph{ramified at $p$} if $B \otimes \Q_p \not\cong M_2(\Q_p)$. It is \emph{ramified at $\infty$} if $B \otimes \R \not\cong M_2(\R)$.
We write $B_{p,\infty}$ for a quaternion algebra (unique up to isomorphism) ramified at $p$ and $\infty$ (and unramified at all other primes).

\begin{lemma}[\cite{pizer_algorithm_1980}]\label{lem:quaternionpq}
Let $p > 2$ be a prime. Then, $B_{p,\infty} \cong \left(\frac{-q,-p}{\Q}\right)$, where
$$q = \begin{cases}
1&\text{if } p \equiv 3 \bmod 4,\\
2&\text{if } p \equiv 5 \bmod 8,\\
q_p&\text{if } p \equiv 1 \bmod 8,\\
\end{cases}$$
where $q_p$ is any prime such that $q_p \equiv 3 \bmod 4$ and $\left(\frac p {q_p}\right) = -1$.
\end{lemma}

Note that if GRH is true, when $p \equiv 1 \bmod 8$, the smallest suitable $q_p$ satisfies $q_p = O((\log p)^2)$ (this follows from~\cite{lagarias_effective_1977}, see also~\cite[Proposition~1]{EC:EHLMP18}).\\

Let $B$ be a quaternion algebra.
The canonical involution $$\overline{a+ib+jc+kd} \longmapsto a-ib-jc-kd$$ induces the \emph{reduced trace} $\Trd(\alpha) =  \alpha + \overline\alpha$ and the \emph{reduced norm} $\Nrd(\alpha) =  \alpha\overline\alpha$.
The reduced norm is a  quadratic form on $B$, with the associated bilinear form
$$\langle \alpha ,\beta \rangle = \frac{1}{2}\Trd(\alpha\overline \beta) = \frac{1}{2}\left(\Nrd(\alpha + \beta) - \Nrd(\alpha) - \Nrd{\beta}\right).$$
When $B$ is ramified at $\infty$, the reduced norm is positive definite.

A \emph{quadratic space} is a $\Q$-vector space $V$ of finite dimension $d$ together with a (positive definite) quadratic form $f : V \to \Q$.
A \emph{lattice} in $V$ is a subgroup $\Lambda \subset V$ of rank $d$ such that $V = \Q \Lambda$.
Given a $\Z$-basis $(b_i)_{i=1}^d$ of a lattice $\Lambda$, its \emph{Gram matrix} is $G = (\langle b_i ,b_j \rangle)_{i,j=1}^d$. The \emph{volume} of the lattice is $\mathrm{Vol}(\Lambda) = \sqrt{|\det(G)|}$, where $G$ is the Gram matrix of any basis of $\Lambda$.  The \emph{discriminant} of $\Lambda$ is $\disc(\Lambda) = 16 \mathrm{Vol}(\Lambda)^2$.

An \emph{order} in a quaternion algebra $B$ is a subring $\mathcal O \subset B$ that is also a lattice. 
It is a \emph{maximal order} if it is not contained in any other order. 
Maximal orders in $B_{p,\infty}$ have $\mathrm{Vol}(\mathcal O) = p/4$.

The following proposition shows that given maximal orders, one can efficiently compute an isomorphism between different models of $B_{p,\infty}$.
\begin{proposition}{\cite[Proposition 4.1]{csahok_explicit_2022}}
\label{prop:quaternion-isomorphism}
Let $A,B$ be quaternion algebras isomorphic to $B_{p,\infty}$, each given by their basis and multiplication table.
Given $\mathcal{O}_A$ a maximal order in $A$ and $\mathcal{O}_B$ a maximal order in $B$, one can compute an isomorphism between $A$ and $B$ in polynomial time.
\end{proposition}

For any lattice $\Lambda \subset B$, its \emph{left order} and \emph{right order} are the orders
$$\mathcal O_L(\Lambda) = \{\alpha \in B \mid \alpha \Lambda \subseteq \Lambda\},\ \ \ \text{and } \mathcal O_R(\Lambda) = \{\alpha \in B \mid  \Lambda\alpha \subseteq \Lambda\}.$$
If $\mathcal O$ is a maximal order, and $I$ is a left ideal in $\mathcal O$, then $\mathcal O_L(I) = I$, and $\mathcal O_R(I)$ is another maximal order. The \emph{connecting ideal} of two maximal order $\mathcal O_1$ and $\mathcal O_2$ is the lattice
$$I(\mathcal O_1,\mathcal O_2) = \{\alpha \in B \mid \alpha \mathcal O_2 \overline \alpha \subseteq [\mathcal O_2 : \mathcal O_1 \cap \mathcal O_2]\mathcal O_1\}.$$
It is a left $\mathcal O_1$-ideal and a right $\mathcal O_2$-ideal.

Let $I$ be a left ideal in a maximal order $\mathcal O$. Its \emph{reduced norm} is $\Nrd(I) = \gcd(\Nrd(\alpha) \mid \alpha \in I) = \sqrt{\#(\mathcal O/I)}$. Its \emph{normalized quadratic form} is the integral quadratic form
$$q_I : I \longrightarrow \Z : \alpha \longmapsto \frac{\Nrd(\alpha)}{\Nrd(I)}.$$
In $B_{p,\infty}$, the volume of $I$ with respect to this quadratic form is $p/4$.
\end{subsection}

\begin{subsection}{Elliptic curves}
See \cite{silverman_arithmetic_2009} for a detailed reference on elliptic curves.
An \emph{elliptic curve} is an abelian variety of dimension $1$.
Given a field $k$ of characteristic $p>3$, an elliptic curve $E$ can be described by a short Weierstrass equation $y^2 = x^3 + ax + b$ for $a,b \in k$ with $4a^3 + 27b^2 \neq 0$. The \emph{$k$-rational points} of $E$ is the set $E(k)$ of pairs $(x,y) \in k^2$ satisfying the curve equation, together with a point $\infty_E$ `at infinity'. They form an abelian group, written additively, where $\infty_E$ is the neutral element.

Given two elliptic curves $E_1$ and $E_2$ defined over $k$, an isogeny $\varphi : E_1 \to E_2$ is a non-constant rational map which is also a group homomorphism from $E_1(\overline k)$ to $E_2(\overline k)$. The \emph{kernel} $\ker(\varphi)$ is a finite subgroup of $E_1(\overline k)$.
The \emph{degree} $\deg(\varphi)$ is its degree as a rational map. An \emph{isomorphism} is an isogeny of degree $1$.

For any integer $m$, the multiplication-by-$m$ map $[m] : E \to E$ is an isogeny. For any isogeny $\varphi : E \to E'$, its \emph{dual} is the unique isogeny $\hat \varphi : E' \to E$ such that $\hat \varphi \circ \varphi = [\deg(\varphi)]$.

The isogeny $\varphi$ is \emph{separable} if $\deg(\varphi) = \#\ker(\varphi)$, and \emph{inseparable} otherwise.
The isogeny is \emph{purely inseparable} if $\ker(\varphi)$ is trivial (note that an isogeny is an isomorphism if and only if it is both separable and purely inseparable).
For any finite subgroup $G \subset E(\overline k)$, there exists a separable isogeny $\varphi : E \to E/G$ with $\ker(\varphi) = G$ (and $\varphi$ is unique up to an isomorphism of the target).

If $k$ has characteristic $p>0$, the map $\phi^E_{p^n} : E \to E^{(p^n)} : (x,y) \mapsto (x^{p^n},y^{p^n})$ is the \emph{$p^n$-Frobenius isogeny}. For any isogeny $\varphi : E\to E'$, there is a maximal integer $n$ such that $\varphi$ factors as $\varphi = \psi \circ \phi^E_{p^n}$. Then, $\psi$ is separable. The isogenies $\phi^E_{p^n}$ are \emph{purely inseparable}.

If $\varphi$ is separable and $\deg(\varphi) = \ell$ is prime, we say that $\varphi$ is an $\ell$-isogeny.\\

An \emph{endomorphism} of $E$ is an isogeny $E \to E$. Endomorphisms, together with the zero morphism, form the \emph{endomorphism ring} $\End(E)$. The curve $E$ is \emph{supersingular} when $\End(E)$ is a lattice of rank $4$. Then, $\End(E) \otimes \Q$ is isomorphic to the quaternion algebra $B_{p,\infty}$, and $\End(E)$ is a maximal order. The degree map $\deg : \End(E) \to \Z$ coincides with the reduced norm of the algebra.

We write $\SSp$ for the set of $\overline{\F}_p$-isomorphism classes of supersingular elliptic curves over $\overline{\F}_p$.
We have $\#\SSp = \lfloor p/12\rfloor + \varepsilon$, with $\varepsilon \in \{0,1,2\}$, and all supersingular elliptic curves have a model over $\F_{p^2}$.

The Deuring correspondence highlights a deep connection between supersingular elliptic curves and maximal orders in the quaternion algebra $B_{p,\infty}$. Indeed, for any supersingular $E/\F_{p^2}$, we have that $\End(E)$ is a maximal order in the algebra $\End(E) \otimes \Q \simeq B_{p,\infty}$. For any isogeny $\varphi : E \to E'$, we have a left ideal 
\[I_\varphi = \Hom(E',E)\circ \varphi \subseteq \End(E),\] 
and the right order $\mathcal O_R(I_\varphi)$ is isomorphic to $\End(E')$. Any non-zero left ideal $I$ in $\End(E)$ is of this form, and we write $\varphi_I : E \to E_I$ for the corresponding isogeny.
We thus have a bijection $\varphi \mapsto I_\varphi$ (with inverse $I \mapsto \varphi_I$) between
\begin{itemize}
\item Isogenies from $E$ (up to isomorphism of the target), and
\item Left ideals in $\End(E)$.
\end{itemize}
Furthermore, for any $\varphi : E\to E'$, the lattice $I_\varphi$ (with its normalized quadratic form $q_{I_\varphi}$) is isomorphic to $\Hom(E,E')$ (with the quadratic form $\deg$).
\end{subsection}

\begin{subsection}{Isogeny algorithms}
There are several ways to encode an isogeny. In computational questions involving isogenies, we typically do not care how an isogeny is encoded, at long as it is an \emph{efficient representation}: an encoding that allows to store and evaluate the isogeny $\varphi$ in polynomial time in $\log(\deg(\varphi))$.

\begin{definition}[\protect{Efficient representation \cite[Definition~1.3]{wesolowski_random_2024}}]\label{def:efficient-representation}
Let $\mathcal A$ be a polynomial time algorithm. It is an \emph{efficient isogeny evaluator} if for any $D \in \{0,1\}^*$ such that $\mathcal A(\mathtt{validity}, D)$ outputs~$\top$,
there exists an isogeny $\varphi: E \to E'$ (defined over some finite field $\F_q$) such that:
\begin{enumerate}
\item on input $(\mathtt{curves}, D)$, $\mathcal A$ returns $(E,E')$,
\item on input $(\mathtt{degree}, D)$, $\mathcal A$ returns $\deg(\varphi)$,
\item on input $(\mathtt{eval},D,P)$ with $P\in E(\F_{q^k})$, $\mathcal A$ returns $\varphi(P)$.
\end{enumerate}
If furthermore $D$ is of polynomial size in $\log(\deg \varphi)$ and $\log q$, then $D$ is an \emph{efficient representation} of $\varphi$ (with respect to $\mathcal A$).
\end{definition}

The break of SIDH \cite{EC:CasDec23,EC:MMPPW23,EC:Robert23} had a major consequence on the computation of isogenies: they can be interpolated. More precisely, one can compute an efficient representation of an isogeny given only its image on a sufficiently large subgroup of its domain. In this paper, we will use the following simplified version of this result.

\begin{proposition}[{\algoname{IsogenyInterpolation} \cite[Theorem 5.19]{EPRINT:Robert24c}}]
\label{prop:hd-rep}
Let $\varphi : E \to E'$ be an $n$-isogeny between supersingular elliptic curves defined over $\mathbb{F}_{p^2}$.
Let $N > n$ be an integer coprime to $pn$ with prime power decomposition $\prod_{i = 1}^r \ell_r^{e_r}$.
Let $(P_1,Q_1,\dots,P_r,Q_r)$ be a set of generators of the $N$-torsion $E[N]$ such that $(P_i,Q_i)$ is a basis of $E[\ell_i^{e_i}]$, for $i = 1,\dots,r$.

Then, given $(n,P_1,Q_1,\dots,P_r,Q_r,\varphi(P_1),\varphi(Q_1),\dots,\varphi(P_r),\varphi(Q_r))$, one can compute an efficient representation of $\varphi$ in polynomial time in the length of the input and in the largest prime factor of $N$. 
\end{proposition}

Using this interpolation result, it was then proved that there is an efficient algorithm to divide isogenies.
This application was first presented in \cite{EPRINT:Robert22c} is some particular case, and later generalised in \cite{hlm_supersingular_2025}.

\begin{proposition}[{\algoname{IsogenyDivision}, \cite{EPRINT:Robert22c} and \cite[Theorem 3]{hlm_supersingular_2025}}]
\label{prop:isogeny-division}
Given an isogeny $\varphi: E_1 \rightarrow E_2$, where $E_1$ and $E_2$ are supersingular elliptic curves defined over $\mathbb{F}_{p^2}$, and an integer $n < \deg \varphi$, one can return an efficient representation of $\varphi/n$ if it is a well-defined isogeny, and return \texttt{False} otherwise, in time polynomial in $\log p$ and $\log \deg \varphi$.
\end{proposition}

Finally, we will use a recent unconditional polynomial-time algorithm to translate $\mathcal{O}$-left ideals to corresponding isogenies, where $\mathcal{O}$ is a maximal order in a quaternion algebra isomorphic to $B_{p,\infty}$.
It is a direct generalisation of the \textsc{Clapoti} algorithm \cite{EPRINT:PagRob23} for computing class group action on oriented elliptic curves in polynomial time, even if the norm of the acting ideal is not smooth.

Translating \textsc{Clapoti} into a general \textsc{IdealToIsogeny} algorithm has been done in prior work, such as \cite{AC:BDDLMP24}. 
Notice that, for the sake of efficiency, the authors of \cite{AC:BDDLMP24} chose to use isogenies in dimension 2, which required assuming heuristics. In contrast, to obtain a rigorous polynomial-time algorithm, we consider the most direct generalisation of \textsc{Clapoti}, using abelian varieties of dimension 8.

\begin{proposition}[\algoname{IdealToIsogeny} \cite{EPRINT:PagRob23}]
\label{prop:ideal-to-isogeny}
Let $E$ be a supersingular elliptic curve defined over $\mathbb{F}_{q}$ such that an isomorphism $\varepsilon: \End(E) \simeq \mathcal{O}$ is known, where $\mathcal{O}$ is a maximal order in some quaternion algebra isomorphic to $B_{p,\infty}$.
Given a left $\mathcal{O}$-ideal $I$, one can compute an efficient representation of the isogeny $\varphi_I : E \to E/E[I]$ in polynomial time in the length of the input. 
\end{proposition}

\end{subsection}

\begin{subsection}{Computational problems}
\label{subsec:problem definitions}
We now formally define the computational problems of interest.
Every isogeny in the input or output of these problems is considered to be in efficient representation~\ref{def:efficient-representation}.
The \emph{supersingular isogeny problem} can then be simply expressed as follows.

\begin{problem}[\Isogeny]
Given $E$ and $E'$ two supersingular elliptic curves defined over $\mathbb{F}_{p^2}$, compute an isogeny $\varphi: E \rightarrow E'$.
\end{problem}

Instead of asking for a single isogeny, one can ask for the collection of all isogenies from $E$ to $E'$. As this collection is a lattice, the following problem asks to find a basis of this lattice.

\begin{problem}[\HomModule]
Given two supersingular elliptic curves $E$ and $E'$ defined over $\mathbb{F}_{p^2}$, find four isogenies generating $\Hom(E,E')$ as a $\mathbb{Z}$-module.
\end{problem}

It is often convenient to find an isogeny of a particular form.
An \emph{$\ell$-isogeny path} (of length $n$) from $E$ to $E'$ is a sequence of $\ell$-isogenies $\varphi_{i} : E_i \to E_{i+1}$ such that $E_0 = E$ and $E_n = E'$. Such a path provides an efficient representation of the degree $\ell^n$ composition $E \to E'$. The following version of the isogeny problem appeared as early as \cite{JC:ChaLauGor09}, when no general method was known for the efficient representation of arbitrary isogenies.

\begin{problem}[\lIsogeny]
Given $E$ and $E'$ two supersingular elliptic curves defined over $\mathbb{F}_{p^2}$, compute an $\ell$-isogeny path from $E$ to $E'$.
\end{problem}

It soon appeared that the problem of finding isogenies is closely related to the problem of finding endomorphisms. Again, there are several ways to formalize it. The first asks to find a basis of the endomorphism ring.

\begin{problem}[\EndRing]
Given a supersingular elliptic curve $E$ defined over $\mathbb{F}_{p^2}$, find four endomorphisms generating $\End(E)$ as a $\mathbb{Z}$-module.
\end{problem}

Similarly to isogenies, we could ask to find a single endomorphism instead of the whole ring. But extra care is required: some endomorphisms are always easy to find: any scalar multiplication $[m]$ on $E$ is an endomorphism, forming the subring $\Z \subset \End(E)$.

\begin{problem}[\OneEnd]
Given a supersingular elliptic curve $E$ defined over $\mathbb{F}_{p^2}$, find an endomorphism in $\End(E) \setminus \mathbb{Z}$.
\end{problem}

Now, \emph{computing the endomorphism ring of $E$} could be interpreted in a different way. Instead of finding actual endomorphisms of $E$ as in \EndRing, one could ask for the abstract structure of $\End(E)$. Indeed, this ring is always isomorphic to a maximal order in $B_{p,\infty}$, and determining which is the following problem.

\begin{problem}[\MaxOrder]
Given a supersingular elliptic curve $E$ defined over $\mathbb{F}_{p^2}$, find a quaternion algebra $B = (\frac{-a,-b}{\mathbb{Q}}) \simeq B_{p,\infty}$, together with four quaternions in $B$ generating a maximal order isomorphic to $\End(E)$. 
\end{problem}

Note that in previous work such as~\cite{EC:EHLMP18,FOCS:Wesolowski21}, \MaxOrder\  does not require finding a model $(\frac{-a,-b}{\mathbb{Q}})$ for $B_{p,\infty}$. Instead, they make the implicit assumption that a fixed model of the form $B_{p,\infty} = (\frac{-p,-q}{\mathbb{Q}})$ is used. Since there is no ``canonical'' model of $B_{p,\infty}$, we include the choice of a model in the definition of the problem. We discuss this subtlety in further detail in \Cref{sec:maxorder}, where we prove that it does not actually matter: it is equivalent to the following problem where a model for $B_{p,\infty}$ of the same form as~\cite{EC:EHLMP18,FOCS:Wesolowski21} is provided.

\begin{problem}[$\MaxOrder_\mathcal Q$]
Let $\mathcal Q$ be an algorithm which for any prime number $p$, outputs a prime $q = \mathcal Q(p)$ such that $B_{p,\infty} \simeq (\frac{-p,-q}{\mathbb{Q}})$. The problem $\MaxOrder_\mathcal Q$ is the following.
Given a supersingular elliptic curve $E$ defined over $\mathbb{F}_{p^2}$, find four quaternions in $(\frac{-p,-\mathcal Q(p)}{\mathbb{Q}})$ generating a maximal order isomorphic to $\End(E)$. 
\end{problem}

As the \EndRing\ problem asks to find actual endomorphisms generating $\End(E)$, and the \MaxOrder\ problem asks to find the ``quaternionic'' structure of $\End(E)$, one can combine these two tasks as in the following problem.

\begin{problem}[\MaxEnd]
Given a supersingular elliptic curve $E$ defined over $\mathbb{F}_{p^2}$, find four endomorphisms $(\alpha_i)_{i=1}^4$ generating $\End(E)$ as a $\mathbb{Z}$-module, a quaternion algebra $B = (\frac{-a,-b}{\mathbb{Q}}) \simeq B_{p,\infty}$, and four quaternions $(\beta_i)_{i=1}^4$ in $B$ such that 
\[\End(E) \otimes \Q \longrightarrow B : \alpha_i \longmapsto \beta_i\]
is an isomorphism.
\end{problem}

\end{subsection}

\begin{subsection}{Random walks}
\label{sec:random-walk}
We will consider random processes in the set of supersingular elliptic curves. Let $\C^{\SSp}$ be the set of functions $\SSp \to \C$. We consider two natural distances on $\C^{\SSp}$. First,
$$\dTV(f,g) = \frac{1}{2}\|f-g\|_1 = \frac{1}{2}\sum_{E \in \SSp}|f(E) - g(E)|.$$
When $f$ and $g$ are distributions on $\SSp$, this is known as the \emph{total variation distance}.
Second, we have the scalar product
\[\langle f, g\rangle = \sum_{E \in \SSp}f(E)\overline {g(E)}\#\mathrm{Aut}(E),\]
inducing the norm $\|f\| = \langle f, f\rangle^{1/2}$.
Note that by the Cauchy-Schwarz inequality and Eichler's formula, for any $f,g$, we have
\[\dTV(f,g)  \leq \frac{\|f-g\|}{2}\left(\sum_{E \in \SSp}\frac{1}{\#\mathrm{Aut}(E)}\right)^{1/2}  = \frac{\|f-g\|}{2} \left(\frac{p-1}{24}\right)^{1/2}.\]

\label{subsection:l-isogeny-graph}
Given an elliptic curve $E$, a prime $\ell$ and an integer $k$, a random $\ell$-walk from $E$ of length $k$ is a random sequence $(\varphi_0,\dots,\varphi_{k-1})$ sampled as follows:
\begin{enumerate}
\item Let $E_0 = E$,
\item For each $i$, let $G_i$ be a uniformly random subgroup of order $\ell$ in $E_i$, and let $E_{i+1} = E_{i}/G_i$
\item For each $i$, let $\varphi_i : E_i \to E_{i+1}$ be the quotient isogeny.
\end{enumerate}
The curve $E$ is the \emph{source} of the walk, and the codomain of $\varphi_{k-1}$ is called the \emph{target} of the walk.
Let $N$ be a positive integer with prime factorization 
$N = \prod_{i=1}^t\ell_i^{k_i}$. A random $N$-walk from $E$ is a sequence $(w_i)_{i=1}^t$ where 
\begin{enumerate}
\item The source of $w_1$ is $E$,
\item Each $w_i$ is a random $\ell_i$-walk of length $k_i$, and
\item For each $i$, the target of $w_i$ is the source of $w_{i+1}$.
\end{enumerate}
The \emph{target} of the walk is the target of $w_t$. Note that the walk itself depends on an order of the prime factors of $N$, but the distribution of the target does not.

The following proposition states that random walks rapidly converge to the so-called \emph{stationary distribution}.

\begin{definition}\label{def:stationary}
The \emph{stationary distribution} on $\SSp$ is the probability distribution defined by $\mu(E) = \frac{24}{(p-1)\#\Aut(E)}$.
\end{definition}

\begin{remark}\label{rem:indistinguishable}
For any $p>3$, the quantity $\#\Aut(E)$ is equal to $2$ for all curves $E$ with two exceptions: if $j(E) = 1728$, then $\#\Aut(E) = 4$, and if $j(E) = 0$, then $\#\Aut(E) = 6$.
Therefore, the total variation distance between the uniform distribution and the stationary distribution on $\SSp$ is $O(1/p)$.
In particular, the two distributions are statistically and computationally indistinguishable.
\end{remark}

\begin{proposition}
\label{prop:random-walk}
Let $N$ be a positive integer with prime factorization 
$N = \prod_{i=1}^t\ell_i^{k_i}$.
Let $E$ be a supersingular elliptic curve over $\F_{p^2}$ sampled from a distribution $f$ on $\SSp$.
Let $W_N(f)$ be the distribution of the target of a random $N$-walk from $E$. Then,
\[\|W_N(f) - \mu\| \leq \|f-\mu\|\cdot\prod_{i=1}^t\left(\frac{2\sqrt{\ell_i}}{\ell_i+1}\right)^{k_i},\]
where $\mu$ is the stationary distribution.
\end{proposition}

\begin{proof}
This is a folklore consequence of Pizer's proof that the supersingular $\ell$-isogeny graph is Ramanujan~\cite{pizer_algorithm_1980}. However, previous literature only details the case where $N$ is a prime power, so let us show that it extends to the general case.
Following \cite[Appendix A.1]{EC:PagWes24}, let $W_\ell = B_\ell/(\ell+1)$ be the $\ell$-walk operator in $X$: for any distribution $f$ on $\SSp$, we have that $W_\ell^k(f)$ is the distribution of the target of a random $\ell$-walk of length $k$. From \cite[Appendix A.1]{EC:PagWes24} and \cite[Theorem 3.10]{EC:PagWes24}, we have $\|W_\ell^k(f) - \mu\| \leq \frac{2\sqrt{\ell}}{\ell+1}\|f-\mu\|$. We deduce that
\[\|W_N(f) - \mu\| = \|(W_{\ell_1}^{k_1}\circ \dots \circ W_{\ell_t}^{k_t})(f) - \mu\| \leq \|f-\mu\|\cdot\prod_{i=1}^t\left(\frac{2\sqrt{\ell_i}}{\ell_i+1}\right)^{k_i},\]
as claimed.
\qed\end{proof}

\begin{remark_star}
	A similar result can be found in \cite[Theorem 1]{AC:GalPetSil17}. However, their limit distribution is incorrect due to the presence of elliptic curves with non-trivial automorphism groups.
\end{remark_star}

In our applications, we only need the following corollary, where we introduce a useful notation $\tau(p,\varepsilon)$ for the subsequent proofs.

\begin{corollary}
\label{cor_randomwalk}
Let $E$ be a random supersingular elliptic curve defined over $\mathbb{F}_{p^2}$, for some distribution $f$, and let $\varepsilon > 0$.
There exists a bound $\tau(p,\varepsilon) = O(\log(p) - \log(\varepsilon))$ such that for any $N > 2^{\tau(p,\varepsilon)}$, the output distribution of a random $N$-walk is at total variation distance at most $\varepsilon$ to the stationary distribution.
\end{corollary}

\begin{proof}
By \cite[Theorem~7, Item~5]{EC:BCCDFL23}, we have $\|f-\mu\|\leq\sqrt{3}$.
Let $\lambda = \log(3/(2\sqrt{2}))$.
From \Cref{prop:random-walk}, we have
\[\|W_N(f)-\mu\| \leq \|f-\mu\|\cdot \prod_{i=1}^t\left(\frac{2\sqrt{\ell_i}}{\ell_i+1}\right)^{k_i}\leq \sqrt{3}\cdot\left(\frac{2\sqrt{2}}{3}\right)^{\log(N)} = \sqrt{3}\cdot2^{-\lambda\log(N)}.\]
Now,
\[\dTV(W_N(f),\mu) \leq \frac{\|W_N(f)-g\|}{2} \left(\frac{p-1}{24}\right)^{1/2} \leq 2^{-\lambda\log(N)} \left(\frac{p-1}{24}\right)^{1/2}.\]
The latter quantity is smaller than $\varepsilon$ if and only if
\[\log(N) \geq \frac{1}{2\lambda} \left(\log(p-1) - \log(24) - 2\log(\varepsilon)\right) = O(\log(p) - \log(\varepsilon)),\]
which proves the result.
\qed\end{proof}

\end{subsection}

Thanks to the concepts introduced in this section, we can now properly define \emph{average-case} problems.

\begin{definition}\label{def:averagecase}
Let $P$ be a problem from the list \lIsogeny, \Isogeny, \EndRing, \OneEnd, \MaxEnd, \MaxOrder, $\MaxOrder_\mathcal{Q}$ and \HomModule.
The input of the problem $P$ consists of one or two supersingular elliptic curves defined over $\F_{p^2}$.
We define the \emph{average-case} for $P$ as the case where the input curves are drawn from the stationary distribution on $\SSp$.
\end{definition}

\end{section}

\begin{section}{The Maximal Order problem}\label{sec:maxorder}

\noindent In this section, we discuss the \MaxOrder\ problem, and a subtlety in its definition when one does not assume GRH.
The classical definition makes the implicit assumption that a reference quaternion algebra $B_{p,\infty}$ is provided. However, there is no ``canonical'' model of $B_{p,\infty}$. When $p \equiv 3 \bmod 4$ (respectively $p \equiv 5 \bmod 8$), one can argue that the algebra $(\frac{-p,-1}{\mathbb{Q}})$ (respectively $(\frac{-p,-2}{\mathbb{Q}})$) is a natural model for $B_{p,\infty}$. However, when $p \equiv 1 \bmod 8$, there is no uniform value of $q$ for which $B_{p,\infty} \simeq (\frac{-p,-q}{\mathbb{Q}})$. 
In order to fix an algebra for each $p$, previous works fix a procedure $\mathcal Q$ such that on input $p$, the output $\mathcal Q(p)$ is a prime satisfying $B_{p,\infty} \simeq (\frac{-p,-\mathcal Q(p)}{\mathbb{Q}})$. This $\mathcal Q(p)$ is typically set to be the smallest prime number with the requested property. While a convenient choice, it is somewhat arbitrary. Furthermore, without GRH, there is no guarantee for this value to be small, nor easy to find.\\

As this model $B_{p,\infty} = (\frac{-a,-b}{\mathbb{Q}})$ \emph{might} be hard to compute (without GRH, when $p \equiv 1 \bmod 8$), the original definition of \MaxOrder\ becomes ambiguous: are $a$ and $b$ provided, or are they to be computed?
We settle for a definition of \MaxOrder\ where $a$ and $b$ are left to be found.
Let us show that the impact of this choice is minimal: it is equivalent to the variant $\MaxOrder_\mathcal Q$ where the algebra is imposed to be of the classical form $(\frac{-p,-\mathcal Q(p)}{\mathbb{Q}})$, for any procedure $\mathcal Q$ which returns a suitable prime.
The key is Proposition~\ref{prop:quaternion-isomorphism}, which allows one to translate solutions accross different models of $B_{p,\infty}$, so the choice of a particular model does not matter.

\begin{proposition}[$\MaxOrder_\mathcal{Q} \text{ is equivalent to } \MaxOrder$]\label{prop:maxorder-equiv-maxorderQ}
Given oracle access to $\mathcal{Q}$, the two problems $\MaxOrder$ and $\MaxOrder_\mathcal Q$ are equivalent under probabilistic polynomial time reductions. The reductions make a single query to each oracle.
\end{proposition}

\begin{proof}
Let $E/\F_{p^2}$ be a supersingular elliptic curve. 
We first prove that the \MaxOrder\ problem reduces to the $\MaxOrder_\mathcal Q$ problem. Indeed, if $\mathcal O \subset (\frac{-p,-\mathcal Q(p)}{\mathbb{Q}})$ is a solution of $\MaxOrder_\mathcal Q$, then $(p,\mathcal Q(p),\mathcal O)$ is a solution of \MaxOrder.

We now prove that $\MaxOrder_\mathcal Q$ reduces to \MaxOrder. Let $\mathcal O \subseteq (\frac{-a,-b}{\mathbb{Q}})$ be a solution of $\MaxOrder$, and let $q = \mathcal Q(p)$. 
Let $\Lambda_0$ be the (non-maximal) order spanned by the canonical basis $(1,i,j,k)$ of $(\frac{-p,-q}{\mathbb{Q}})$.
	Thanks to the orthogonality of this basis and since the discriminant of any order $\mathcal{O} = (\alpha_1,\dots,\alpha_4)$ is given by $\disc(\mathcal{O}) = \sqrt{|\det((\langle \alpha_i,\alpha_j \rangle)_{i,j})|}$, the discriminant of $\Lambda_0$ is $pq$. 
	The factorisation of $\disc(\Lambda_0)$ being known, one can construct a maximal order $\mathcal{O}_0 \supseteq \Lambda_0$ in polynomial time with~\cite[Theorem~7.14]{voight_identifying_2013}.

From Proposition~\ref{prop:quaternion-isomorphism}, one can compute an order $\mathcal O'$ in $(\frac{-p,-q}{\mathbb{Q}})$ isomorphic to $\mathcal O$. Then, $\mathcal O'$ is a solution of $\MaxOrder_\mathcal Q$.
\qed\end{proof}

\end{section}

\begin{section}{The Homomorphism Module problem}\label{sec:hommod}

The $\HomModule$ problem has not been formally studied in the previous literature, yet it naturally appears in isogeny-based cryptography.
For instance, the homomorphism module between the commitment curve and the challenge curve in SQISign \cite{AC:DKLPW20} is the space of all possible responses during an identification.
While it is clear that $\Isogeny$ reduces to $\HomModule$, we prove in this section that both are actually equivalent.

The relation between \HomModule\ and \Isogeny\ is reminiscent of the relations between \EndRing\ and \OneEnd.
The latter equivalence has been proved in \cite{EC:PagWes24}.
In this same paper, the authors also proved that \OneEnd\ reduces to \Isogeny, both of these reductions being unconditional.
Therefore, there is a probabilistic polynomial time algorithm solving \EndRing\ given an \Isogeny\ oracle.
In order to take advantage of this fact to solve \HomModule\, we prove that knowing an isogeny between two elliptic curves and their respective endomorphism rings, one can compute efficiently a basis of the homomorphism module between the said curves, \Cref{prop:two-endrings-plus-inseparable-isogeny} and \Cref{prop:two-endrings-plus-isogeny}.
This leads to the main result of the section, \Cref{prop:HomModule-to-Isogeny}, proving that \HomModule\ reduces to \Isogeny.

\begin{lemma}\label{lemma-isogenies-coprime}
Let $\varphi_i : E \to E_i$, for $i \in \{1,2\}$, be separable isogenies such that $\ker \varphi_1 \cap \ker \varphi_2 = 0$. Then,
$$\Hom(E_1,E)\varphi_1 + \Hom(E_2,E)\varphi_2 = \End(E).$$
\end{lemma}
\begin{proof}
Let $\varphi_3 : E \to E_3$ be a separable isogeny with $\ker \varphi_3 = \ker \varphi_1 + \ker \varphi_2$. Since $\ker \varphi_1 \cap \ker \varphi_2 = 0$, we have $|\ker \varphi_3| = |\ker \varphi_1||\ker \varphi_2|$.
Each 
$$I_i = \Hom(E_i,E)\varphi_i = \{\alpha \in \End(E) \mid \ker \varphi_i \subseteq \ker \alpha \} $$
is a left $\End(E)$-ideal of reduced norm $\Nrd(I_i) = |\ker\varphi_i|$.
We have
\[I_1 \cap I_2 = \{\alpha \in \End(E) \mid (\ker \varphi_1 + \ker \varphi_2) \subseteq \ker \alpha \} = I_3.\]
We have, 
\begin{align*}
|I_1/(I_1 \cap I_2)|^{1/2} &= \frac{\Nrd(I_3)}{\Nrd(I_1)} = \frac{|\ker\varphi_3|}{|\ker\varphi_1|} = \frac{|\ker\varphi_1||\ker\varphi_2|}{|\ker\varphi_1|} \\
&= \Nrd(I_2) = |\End(E)/I_2|^{1/2}.
\end{align*}
By the \emph{second isomorphism theorem},
\[I_1/(I_1 \cap I_2) \cong (I_1+I_2)/I_2 \subseteq \End(E)/I_2,\]
and since the leftmost and rightmost quotients have the same cardinality, we deduce $(I_1+I_2)/I_2 = \End(E)/I_2$, hence $I_1+I_2 = \End(E)$.
\qed\end{proof}

\begin{lemma}\label{lemma:two-isogenies-generate-all-isogenies}
Let $\varphi_i : E \to E_i$, for $i \in \{1,2\}$, be separable isogenies such that $\ker \varphi_1 \cap \ker \varphi_2 = 0$. Then,
for any elliptic curve $E'$,
$$\Hom(E_1,E')\varphi_1 + \Hom(E_2,E')\varphi_2 = \Hom(E,E').$$
\end{lemma}
\begin{proof}
Clearly $\Hom(E_1,E')\varphi_1 + \Hom(E_2,E')\varphi_2  \subseteq \Hom(E,E')$, so let us prove the second inclusion.
By \Cref{lemma-isogenies-coprime},
\begin{align*}
\Hom(E,E') &= \Hom(E,E')\End(E) \\
&= \Hom(E,E')(\Hom(E_1,E)\varphi_1 + \Hom(E_2,E)\varphi_2)\\
&\subseteq \Hom(E_1,E')\varphi_1 + \Hom(E_2,E')\varphi_2,
\end{align*}
which proves the result.
\qed\end{proof}

\begin{proposition}\label{prop:two-endrings-plus-separable-isogeny}
Let $\varphi : E \to E'$ be a separable isogeny. Then,
$$\mathrm{span}_\mathbb Z(\End(E')\varphi\End(E)) = m\Hom(E,E'),$$
where $m\in\mathbb Z$ is the largest integer dividing $\varphi$.
\end{proposition}

\begin{proof}
Clearly $\mathrm{span}_\mathbb Z(\End(E')\varphi\End(E)) \subseteq m\Hom(E,E')$, so let us prove the other inclusion.
Write $\varphi = m\psi$ with $\ker \psi$ cyclic. 
Let $n = \deg(\psi)$. The kernel $\ker \psi \cong \mathbb Z/n\mathbb Z$ is a cyclic subgroup of $E[n]\cong (\mathbb Z/n\mathbb Z)^2$. The action of $\End(E)/n\End(E)$ on $E[n]$ is isomorphic to the action of $M_2(\mathbb Z/n\mathbb Z)$ on $(\mathbb Z/n\mathbb Z)^2$, so there exists an endomorphism $\alpha \in \End(E)$ (of degree coprime with $n$) such that $\ker (\psi) \cap \alpha^{-1}(\ker \psi) = 0$.
In other words, $\ker (\psi) \cap \ker(\psi\alpha) = 0$. Applying Lemma~\ref{lemma:two-isogenies-generate-all-isogenies}, we deduce
$$\End(E')\psi + \End(E')\psi\alpha = \Hom(E,E').$$
We deduce
$$m\Hom(E,E') = \End(E')\varphi + \End(E')\varphi\alpha \subseteq \mathrm{span}_\mathbb Z(\End(E')\varphi\End(E)),$$
which proves the proposition.
\qed\end{proof}

Recall that any isogeny can be factored as $\phi\varphi$ where $\varphi$ is separable, and $\phi$ is purely inseparable ($\phi$ might be an isomorphism).
Then, the following proposition generalized \Cref{prop:two-endrings-plus-separable-isogeny} to arbitrary isogenies.

\begin{proposition}\label{prop:two-endrings-plus-inseparable-isogeny}
Let $\varphi : E \to E''$ be a separable isogeny, and $\phi : E'' \to E'$ a purely inseparable isogeny. Then,
$$L = \mathrm{span}_\mathbb Z(\End(E')\phi\varphi\End(E)) = m\phi\Hom(E,E''),$$
where $m\in\mathbb Z$ is the largest integer dividing $\varphi$.
\end{proposition}

\begin{proof}
Since $\phi : E'' \to E'$ is purely inseparable, we have $\End(E')\phi = \phi\End(E'')$.
The result then immediately follows from \Cref{prop:two-endrings-plus-separable-isogeny}.
\qed\end{proof}

\begin{proposition}\label{prop:two-endrings-plus-isogeny}
Let $E$, $E'$ and $E''$ be supersingular elliptic curves, and $\phi : E'' \to E'$ a purely inseparable isogeny. 
Given a basis of $\phi\Hom(E,E'')$, one can compute a basis of $\Hom(E,E')$ in polynomial time.
\end{proposition}

\begin{proof}
Let $(b_i)_{i=1}^4$ be the provided basis of the lattice $L = \phi\Hom(E,E'')$.
Let $p^n = \deg(\phi)$.
If $n = 2m$ is even, then $\phi = p^m\alpha$ where $\alpha : E'' \to E'$ is an isomorphism. Then $(b_i/p^m)_{i=1}^4$ is a basis of $\alpha\Hom(E,E'') = \Hom(E,E')$.
If $n = 2m + 1$ is odd, one can similarly divide by $p^m$, and without loss of generality we now consider the case where $\phi$ is the $p$-Frobenius.
Consider the quadratic form 
$$q : L \longrightarrow \mathbb Z : \varphi \longmapsto \deg(\varphi)/p.$$
We have
$$X := p\Hom(E,E') = L \cap (p\Hom(E,E')) = \{\varphi\in L \mid q(\varphi) \equiv 0 \bmod p\}.$$
The equation $q(\varphi) \equiv 0 \bmod p$ thus defines an $\mathbb F_p$-linear subspace of $L/pL$. Let us prove that it is equal to (hence can be computed as) the kernel of the Gram matrix over $\mathbb F_p$.
Indeed, the kernel of the Gram matrix over $\mathbb F_p$ is 
\[K = \{\varphi\in L \mid \langle\varphi,\psi\rangle_q \equiv 0 \bmod p\text{ for all }\psi \in L\},\]
with 
\[\langle\varphi,\psi\rangle_q = \frac{1}{2}(q(\varphi + \psi) - q(\varphi) - q(\psi)) = \frac{1}{2p}(\hat \varphi \psi + \hat \psi \varphi),\]
the bilinear form associated to $q$. Since $q(\varphi) = \langle\varphi,\varphi\rangle_q$, we have $K \subseteq X$. 
It only remains to prove that $X \subseteq K$. Let $\varphi = [p]\varphi' \in X$. For any $\psi = \phi\psi' \in L$ we have
\[\langle\varphi,\psi\rangle_q = \frac{1}{2p}([p]\hat \varphi' \phi\psi' + \hat \phi\hat\psi' [p]\varphi') = \frac{1}{2}(\hat \varphi' \phi\psi' + \hat \phi\hat\psi' \varphi').\]
Since both $\phi$ and $\hat \phi$ are inseparable, the (scalar) endomorphism $\hat \varphi' \phi\psi' + \hat \phi\hat\psi' \varphi'$ is inseparable, so it is $0$ modulo $p$. It holds for all $\psi$, so $\varphi \in X$, hence $X = K$.
\qed\end{proof}

\begin{algorithm}
\caption{Reducing \HomModule\ to \Isogeny} \label{alg_redHom2Iso}
\begin{flushleft}
\hspace*{\algorithmicindent} \textbf{ Input: } Two supersingular elliptic curves $E_1$ and $E_2$ and an access to an oracle of \Isogeny.\\
	\hspace*{\algorithmicindent} \textbf{ Output: } Four isogenies $\varphi_i : E_1 \rightarrow E_2$, $i \in \{ 1,...,4 \}$ generating $\Hom(E_1,E_2)$ as a $\Z$-module.
\end{flushleft}
\begin{algorithmic}[1]
	\State $( \alpha_1, \alpha_2, \alpha_3, \alpha_4 ) \gets$ a basis of $\End(E_1)$; \Comment{\cite[Theorem 8.6]{EC:PagWes24}}
	\State $( \beta_1, \beta_2, \beta_3, \beta_4 ) \gets$ a basis of $\End(E_2)$; \Comment{\cite[Theorem 8.6]{EC:PagWes24}}
\State Compute an isogeny $\varphi: E_1 \rightarrow E_2$; \Comment{Using the \Isogeny{} oracle} \label{step:isogeny-from-oracle}
	\State $v \gets v_p(\deg \varphi)$; \Comment{$p$-adic valuation of $\deg \varphi$}
\State $S \gets \{\beta_j \circ \varphi \circ \alpha_i\} \subset \Hom(E_1,E_2)$;
	\State $(\gamma_1, \gamma_2, \gamma_3, \gamma_4) \gets$ a basis of the lattice generated by $S$; \Comment{\cite{buchmann_computing_1989}}
\State $m \gets \left(16 \det(\langle \gamma_i, \gamma_j \rangle)/p^{4v +2}\right)^{1/8}$; \label{step:compute-m}
\State $B_0 \gets (\gamma_1/m, \gamma_2/m, \gamma_3/m, \gamma_4/m)$; \Comment{\Cref{prop:isogeny-division}} \label{step:compute-B}
\State Extract a basis of $\Hom(E,E')$ from $B$; \Comment{\Cref{prop:two-endrings-plus-isogeny}}

\State \Return $B$.

\end{algorithmic}
\end{algorithm}

\begin{proposition}[$\HomModule \text{ reduces to } \Isogeny$]
\label{prop:HomModule-to-Isogeny}
Algorithm \ref{alg_redHom2Iso} is correct and runs in time polynomial in $\log p$ and in the length of the oracle outputs.
\end{proposition}

\begin{proof}
The running time of each step is ensured by the corresponding reference.
In particular, each is polynomial in $\log p$ and in the length of the oracle's outputs.

To prove the correctness of the algorithm, we begin by writing the isogeny $\varphi$, returned by the \Isogeny\ oracle at Step~\ref{step:isogeny-from-oracle}, as $\varphi = \phi \circ \psi$, where $\psi: E_1 \rightarrow E'$ is a separable isogeny and $\phi: E' \rightarrow E_2$ is a purely inseparable isogeny of degree $p^v \geq 1$, with $v := v_p(\deg \varphi)$.
Note that when $v = 0$, the purely inseparable part $\phi$ is an isomorphism.
Let $m$ be the largest integer that divides $\psi$. 
Then, by Proposition~\ref{prop:two-endrings-plus-inseparable-isogeny}, the set $S$ generates the lattice $m\phi \Hom(E_1,E')$.
For a basis $(\gamma_1,\dots,\gamma_4)$ of the lattice generated by $S$, we have
\begin{align*}
	\det(\langle \gamma_i, \gamma_j \rangle) &= \mathrm{Vol}(m\phi\Hom(E_1,E'))^2\\
	&=  (m^4 \deg(\phi)^2\mathrm{Vol}(\Hom(E_1,E')))^2\\ 
	&= m^8 p^{4v + 2}/16,
\end{align*}
thus the computation at Step~\ref{step:compute-m} gives the correct $m$.
In particular, the basis $B = (\gamma_1/m,\dots,\gamma_4/m)$ generates $\phi \Hom(E_1,E')$.
From it, one can compute a basis of $\Hom(E_1,E_2)$ using Proposition~\ref{prop:two-endrings-plus-isogeny}.
\qed\end{proof}

\end{section}

\begin{section}{Finding endomorphisms from quaternions}\label{sec:oneend to maxorder}
In the section we develop an unconditional reduction from the $\OneEnd$ problem to the $\MaxOrder$ problem.
The main difficulty is that without GRH, there is no general way to compute a ``special''
elliptic curve $E_0$ for which both $\End(E_0)$ and its embedding in the quaternions are already known.
Such a curve provides an ``endomorphism/quaternion'' dictionary, and
previous literature on $\MaxOrder$ made critical use of that fact.
Without such an $E_0$, we need to develop a completely different strategy.
To reduce $\OneEnd$ (say on some input $E$) to $\MaxOrder$, we solve the $\MaxOrder$ problem on $E$, giving an order $\mathcal O$, but also on a few of its ``neighbours''. We thereby construct a ``local'' correspondence: a canonical bijection between $\ell$-isogenies from $E$ and ideals of norm $\ell$ in the order $\mathcal O$. This is done in \Cref{alg_solveOneEndOrFingIsoIdBijection}. We then prove that this information can be converted into isomorphisms $\End(E[\ell]) \simeq \mathcal O/\ell\mathcal O$ which all descend from the same implicit isomorphism $\End(E) \simeq \mathcal O$, via \Cref{alg_solveOneEndOrFingIso}. Finally, from this local data, we reconstruct a full ``endomorphism/quaternion'' dictionary $\End(E) \simeq \mathcal O$ in \Cref{alg:oneEndtoMaxOrder}.

At several steps of this process, one might fail to construct the dictionary (for instance when the isomorphism $\End(E) \simeq \mathcal O$ is not unique). In such a scenario, a non-scalar endomorphism of $E$ is revealed, and we have solved $\OneEnd$ anyway. If no such failure occurs, we successfully obtain a dictionary, which in turn reveals (all!) non-scalar endomorphisms of $E$.

\begin{algorithm}[h]
\caption{Computing a bijection between $\ell$-isogenies and ideals, given a $\MaxOrder$ oracle} \label{alg_solveOneEndOrFingIsoIdBijection}
\begin{flushleft}
	\hspace*{\algorithmicindent} \textbf{ Input: } A supersingular elliptic curve $E/\F_{p^2}$, a prime $\ell \neq p$, a list $(G_i)_{i=0}^{\ell}$ of all subgroups of order $\ell$ in $E$, an algebra $B = (\frac{-a,-b}{\mathbb{Q}}) \simeq B_{p,\infty}$, a maximal order $\mathcal O \simeq \End(E)$ in $B$, and access to an oracle for \MaxOrder.\\
\hspace*{\algorithmicindent} \textbf{ Output: } Either an endomorphism $\alpha \in \End(E)\setminus \Z$, or the list $(I_i)_{i=0}^{\ell}$ of left $\mathcal O$-ideals such that $\mathcal O_R(I_i) \simeq \End(E/G_j)$ if and only if $i=j$.
\end{flushleft}
\begin{algorithmic}[1]
\State Compute $\varphi_i : E \to E/G_i$ using V\'elu's formulas for $i = 0, \dots, \ell$;
\If {$E/G_i \simeq E/G_j$ for some $i \not = j$} \label{alg_bij_step_if}
	\State $\gamma \gets$ an isomorphism between $E/G_i$ and $E/G_j$;
	\State \Return $\hat \varphi_j \circ \gamma \circ \varphi_i \in \End(E)\setminus \Z$; \label{alg_bij_step_alpha}
\EndIf\If {$E/G_i \simeq (E/G_j)^{(p)}$ for some $i \not = j$} \label{alg_bij_step_if_frob}
	\State $\gamma \gets$ an isomorphism between $E/G_i$ and $(E/G_j)^{(p)}$;
	\State $\phi_p \gets$ the Frobenius isogeny $\phi_p : E/G_j \to (E/G_j)^{(p)}$;
	\State \Return $\hat \varphi_j \circ \phi_p \circ \gamma \circ \varphi_i \in \End(E)\setminus \Z$; \label{alg_bij_step_alpha_frob}
\EndIf
\For {$i = 0,\dots,\ell$} \label{alg_bij_step_for}
	\State $(B_i,\tilde {\mathcal O}_i) \gets$ an algebra $B_i \simeq B_{p,\infty}$ and a maximal order $\tilde {\mathcal O}_i \subset B_i$;
	\Statex \hspace*{\algorithmicindent} such that $\tilde {\mathcal O}_i \simeq \End(E/G_i)$; \Comment{Using the oracle for \MaxOrder\ on $E/G_i$}
	\State $\mathcal O_i \gets$ an order in $B$ isomorphic to $\End(E/G_i)$; \Comment{Proposition~\ref{prop:quaternion-isomorphism} on $(B,\mathcal O)$ \hspace*{7.9cm} and $(B_i,\tilde {\mathcal O}_i)$}
	\State $J_i \gets I(\mathcal O, \mathcal O_i)$ the connecting ideal;\Comment{\cite[Algorithm~3.5]{kirschmer_algorithmic_2010}}
	\State $\langle \alpha_1, \dots, \alpha_4\rangle \gets$ a Minkowski-reduced basis of $J_i$;\Comment{\cite{nguyen_low-dimensional_2009}}
	\If {$\Nrd(\alpha_2) \leq \ell \Nrd(J_i)$}
		\State {$\alpha \gets$ a non-scalar endomorphism of $E$ of degree at most $\ell^2$; \label{alg_bij_step_exhaustive}}
		\Statex \hspace*{\algorithmicindent} \Comment{For instance, by exhaustive enumeration of isogenies of degree at most $\ell^2$ \hspace*{\algorithmicindent} \hspace*{\algorithmicindent} from $E$}
		\State \Return $\alpha$;
	\EndIf
	\State $I_i \gets J_i\overline \alpha_1/\Nrd(J_i)$;\Comment{The unique left $\mathcal O$-ideal of norm $\ell$ equivalent to $J_i$}
\EndFor
\State \Return $(I_i)_{i=0}^{\ell}$. \label{alg_bij_step_final}
\end{algorithmic}
\end{algorithm}

\begin{lemma}\label{lemma:solveOneEndOrFingIsoIdBijection}
Algorithm~\ref{alg_solveOneEndOrFingIsoIdBijection} is correct and runs in time polynomial in the length of the input, in $\ell$, and in the length of the \MaxOrder\ oracle outputs.
\end{lemma}

\begin{proof}
The claim that the running time is polynomial follows from the references provided in the comments of Algorithm~\ref{alg_solveOneEndOrFingIsoIdBijection}.

Let us prove that the algorithm is correct. 

First, the endomorphism $\alpha = \hat \varphi_j \circ \gamma \circ \varphi_i$ returned at Line~\ref{alg_bij_step_alpha} is not scalar. Indeed, suppose by contradiction that $\alpha \in \Z$. It is of degree $\ell^2$, so $\alpha = [\ell]$. We thus have $\hat \varphi_j \circ \varphi_j = [\ell] = \hat \varphi_j \circ \gamma \circ \varphi_i$, hence $ \varphi_j = \gamma \circ \varphi_i$, hence $\ker(\varphi_j) = \ker(\varphi_i)$, hence $G_i = G_j$, contradicting that $i \neq j$.

Second, the endomorphism returned at Line~\ref{alg_bij_step_alpha_frob} is not scalar either. Indeed, it has degree $\ell^2p$, which is not a square, so it cannot be scalar.

Note that after Line~\ref{alg_bij_step_alpha_frob}, the rings $\End(E/G_i)$ are pairwise non-isomorphic. Indeed, by \cite[Lemma 42.4.1]{voight_quaternion_2021}, if $\End(E/G_i)\simeq \End(E/G_j)$, then $E/G_i$ is isomorphic to either $E/G_j$ or to its Galois conjugate $(E/G_j)^{(p)}$, in which case the algorithm has terminated before Line~\ref{alg_bij_step_alpha_frob}.
In particular, the codomains of all $\ell$-isogenies from $E$ have pairwise distinct endomorphism rings, hence left ideals of norm $\ell$ in $\mathcal O$ are uniquely identified by the isomorphism class of their right-order.

At each iteration of the for-loop, we consider two cases.
\begin{itemize}
	\item
	\underline{If $\Nrd(\alpha_2) \leq \ell \Nrd(J_i)$}, by definition of Minkowski bases, we have that $$\Nrd(\alpha_1) \leq \ell \Nrd(J_i).$$
As $J_i \bar J _i = \Nrd(J_i)\mathcal{O}$, the element $\alpha_1 \bar{\alpha}_2 / \Nrd(J_i)$ is in $\mathcal{O}$.
Furthermore, $\alpha_1 \bar{\alpha}_2 / \Nrd(J_i)$ is not a scalar; otherwise, $\alpha_1$ and $\alpha_2$ would be linearly dependent, since $\Nrd(\alpha_2) \alpha_1 = \left(\frac{\alpha_1 \bar{\alpha}_2}{\Nrd(J_i)}\right)\Nrd(J_i) \alpha_2$.
Therefore $\mathcal{O}$ (hence also $\End(E)$) contains a non-scalar element of norm 
\[\Nrd\left(\frac{\alpha_1 \bar{\alpha}_2 }{ \Nrd(J_i)}\right) = \frac{\Nrd\left(\alpha_1  \right)\Nrd\left(\alpha_2 \right)}{\Nrd(J_i)^2} \leq \ell^2.\]
Then, a non-trivial endomorphism of degree at most $\ell^2$ can be found in time polynomial in $\log p$ and $\ell$ by exhaustive search.
	\item
	\underline{Otherwise}, $\Nrd(\alpha_2) > \ell \Nrd(J_i)$. Let us prove that in that case, the ideal $I_i = J_i \bar \alpha_1/\Nrd(J_i)$ is the unique left $\mathcal{O}$-ideal of norm $\ell$ with $\mathcal{O}_R(I_i) \simeq \End(E/G_i)$.
	Recall that by the Deuring correspondence, the lattice $\Hom(E,E/G_i)$ (for the quadratic form $\deg$) is isomorphic to $J_i$ (with quadratic form $q_{J_i} : \alpha \mapsto \Nrd(\alpha)/\Nrd(J_i)$). Therefore, there exists an element $\beta \in J_i$ such that $q_{J_i}(\beta) = \ell$. Since $q_{J_i}(\beta) < q_{J_i}(\alpha_2)$, the element $\beta$ must be a multiple of $\alpha_1$: there exists $m \in \Z$ such that $\beta = m\alpha_1$. Since 
	\[\ell = q_{J_i}(\beta) =q_{J_i} (m\alpha_1) = m^2q_{J_i}(\alpha_1)\] and $\ell$ is prime, we must have that $m = 1$ and $q_{J_i}(\alpha_1) = \ell$. This implies that $\Nrd(I_i) = \ell$.
	
The unicity of $I_i$ follows from the previously established fact that left ideals of norm $\ell$ in $\mathcal O$ are uniquely identified by the isomorphism class of their right-order.
\end{itemize}
The unicity of each $I_i$ proves that if Line~\ref{alg_bij_step_final} is reached, we indeed have $\mathcal O_R(I_i) \simeq \End(E/G_j)$ if and only if $i=j$.
\qed\end{proof}

\begin{lemma}\label{lemma:automorphism-determined-by-ideals}
Let $\rho : M_2(\F_\ell) \to M_2(\F_\ell)$ be a ring automorphism. If $\ker(m) = \ker(\rho(m))$ for all $m \in M_2(\F_\ell)$, then $\rho$ is the identity.
\end{lemma}

\begin{proof}
All automorphisms of $M_2(\F_\ell)$ are inner, so there exists $p \in \GL_2(\F_\ell)$ such that $\rho(m) = p^{-1}mp$ for all $m \in M_2(\F_\ell)$.

First, suppose there exists a line $L \subset \F_\ell^2$ passing through the origin such that $p(L) \neq L$.
Then, there exists $m \in M_2(\F_\ell)$ such that $m(L) = p(L)$ and $m(p(L)) = \{0\}$. We obtain
$$\rho(m)(L) = (p^{-1}mp)(L) = p^{-1}(m(p(L)) = p^{-1}(\{0\}) = \{0\}.$$
By construction, $m$ cannot be invertible or the zero matrix; consequently, both $\ker(\rho(m))$ and $\ker(m)$ have dimension $1$. 
Therefore
$L = \ker(\rho(m)) = \ker(m) = p(L),$
a contradiction. We deduce that for any line $L \subset \F_\ell^2$, we have $p(L) = L$.
Since $p$ fixes all lines in $\F_\ell^2$, all vectors of $\F_\ell^2$ are eigenvectors of $p$, so $p$ is a scalar matrix. In particular, $\rho(m) = p^{-1}mp = m$.
\qed\end{proof}

\begin{corollary}\label{coro:isomorphism-determined-by-ideals}
Consider rings $R \simeq R' \simeq M_2(\F_\ell)$. 
Let $\iota_1,\iota_2 : R \to R'$ be two ring isomorphisms.
If $\iota_1(I) = \iota_2(I)$ for all left-ideals $I$ in $R$, then $\iota_1 = \iota_2$.
\end{corollary}

\begin{proof}
Fix two isomorphisms $g : R' \to M_2(\F_\ell)$ and $f : M_2(\F_\ell) \to R$, and define 
\[\rho_i = g \circ \iota_i \circ f : M_2(\F_\ell) \to M_2(\F_\ell).\]
Let $\rho = \rho_2^{-1}\circ \rho_1$. 
Let us prove that $\rho$ satisfies the condition of Lemma~\ref{lemma:automorphism-determined-by-ideals}. Let $m \in M_2(\F_\ell)$. 
Let $J = M_2(\F_\ell)m = \{\tilde m \in M_2(\F_\ell) \mid \ker(m) \subseteq \ker(\tilde m)\}$ be the left-ideal generated by $m$. Then, its image $f(J)$ is a left-ideal in $R$, hence \(\iota_1 ( f (J)) = \iota_2 ( f (J))\), and
\[\rho(J) = f^{-1} \circ \iota_2^{-1} \circ \iota_1 \circ f (J) = f^{-1} \circ \iota_2^{-1} \circ \iota_2 \circ f (J) = J.\]
In particular, $\rho(m) \in \rho(J) = J$, so $\ker(m) \subseteq \ker(\rho(m))$. Since $\rho$ is an isomorphism, the matrices $m$ and $\rho(m)$ have the same rank, hence $\ker(m) = \ker(\rho(m))$.

We can thus apply Lemma~\ref{lemma:automorphism-determined-by-ideals}, and deduce that $\rho$ is the identity.
In particular, we obtain $\rho_1 = \rho_2$, therefore $\iota_1 = \iota_2$.
\qed\end{proof}

\begin{algorithm}
\caption{Computing an isomorphism between quaternions and endomorphisms modulo $\ell$, given a $\MaxOrder$ oracle} \label{alg_solveOneEndOrFingIso}
\begin{flushleft}
\hspace*{\algorithmicindent} \textbf{ Input: } A supersingular elliptic curve $E/\F_{p^2}$, a prime $\ell$, an algebra $B \simeq B_{p,\infty}$, a maximal order $\mathcal O \simeq \End(E)$ in $B$, and access to an oracle for \MaxOrder.\\
\hspace*{\algorithmicindent} \textbf{ Output: } Either an endomorphism $\alpha \in \End(E)\setminus \Z$, or an isomorphism $\lambda : \mathcal{O}/\ell\mathcal{O} \to \End(E[\ell])$.
\end{flushleft}
\begin{algorithmic}[1]
	\State $(G_i)_{i=0}^\ell \gets$ a list of all subgroups of order $\ell$ of the elliptic curve $E$;
	\State Using the oracle access, run Algorithm~\ref{alg_solveOneEndOrFingIsoIdBijection} on the list $(G_i)_{i=0}^\ell$ to obtain either
	\begin{itemize}
		\item a non trivial endomorphism $\alpha \in \End(E) \setminus \Z$,
		\item or a list $(I_i)_{i=0}^\ell$ such that $I_i$ is the unique left $\mathcal{O}$-ideal of norm $\ell$ with $\mathcal{O}_R(I_i) \simeq \End(E/G_i)$;
	\end{itemize}

	\If{ $\alpha \in \End(E) \setminus \Z $ was found } \Return $\alpha$; \EndIf

	\State $g \gets$ an isomorphism from $\End(E[\ell])$ to $M_2(\F_\ell)$;
	\State $f \gets$ an isomorphism $\mathcal{O}/\ell\mathcal{O}$ to $M_2(\F_\ell)$;
	\State $J_i \gets \{\alpha \in \End(E[\ell])| G_i \subset \ker \alpha \}$, for $i \in \{0, \dots, \ell\}$;
	\State $h \gets$ an automorphism from $M_2(\F_\ell)$ to $M_2(\F_\ell)$ such that $h(f(\tilde{I}_i)) = g(J_i)$ where $\tilde I_i$ is the reduction of $I_i$ modulo $\ell$;
	\State $\lambda \gets g^{-1} \circ h \circ f: \mathcal{O}/\ell\mathcal{O} \to \End(E[\ell])$;
	\State \Return $\lambda$.
\end{algorithmic}
\end{algorithm}

\begin{lemma}\label{lemma:endo-image-on-torsion}
Algorithm~\ref{alg_solveOneEndOrFingIso} is correct and runs in time polynomial in the length of the input, in  $\ell$, and in the length of the output of the oracle for \MaxOrder.
Moreover, when it returns an isomorphism $\lambda: \mathcal{O}/\ell\mathcal{O} \stackrel{\sim}{\rightarrow} \End(E[\ell])$, this isomorphism is the reduction modulo $\ell$ of any (and all) isomorphism $\iota: \mathcal{O} \stackrel{\sim}{\rightarrow} \End(E)$.
\end{lemma}

\begin{proof}
Let $\iota: \mathcal{O} \stackrel{\sim}{\to} \End(E)$ be an isomorphism.
Let us prove that the isomorphism $\lambda$ computed by Algorithm~\ref{alg_solveOneEndOrFingIso} is its reduction modulo $\ell$, thereby also proving the correctness of this algorithm.

Since $I_i$ is the unique left $\mathcal{O}$-ideal of norm $\ell$ with $\mathcal{O}_R(I_i) \simeq \End(E/G_i)$, we have that 
$$\iota(I_i) = \{ \alpha \in \End(E) | G_i \subseteq \ker(\alpha) \}.$$
Then, reduced modulo $\ell$, the equality becomes
$\iota_\ell(\tilde{I_i}) =  J_i,$
where $\tilde{I_i}$ is the reduction of $I_i$ modulo $\ell$.
On the other hand, by construction, we have $\lambda(\tilde{I_i}) = J_i$.
Thus, by Corollary~\ref{coro:isomorphism-determined-by-ideals}, the isomorphism $\iota_\ell$ is equal to $\lambda$.

Now, we demonstraste the complexity of the algorithm by giving the complexity of its different steps.
Obtaining the list of subgroups of order $\ell$ of the elliptic curve $E$ can be done by computing a basis of the $\ell$-torsion of $E$.
This takes a polynomial time in $\ell$ and in $\log p$.

One can define the isomorphisms $g,f$ and $h$ by mapping the basis of the domain to a basis of the codomain such that the map verifies the respective required properties.
Since there are $O(\ell^4)$ ordered bases of $M_2(\F_\ell)$, these constructions can be carried out using an exhaustive search.

Finally, by Lemma~\ref{lemma:solveOneEndOrFingIsoIdBijection}, running Algorithm~\ref{alg_solveOneEndOrFingIsoIdBijection} takes a polynomial time in the length of the input, in $\ell$ and in the length of the output of the oracle for \MaxOrder.
All the previously discussed complexities are encompassed within this running time.

\qed\end{proof}

\begin{algorithm}
\caption{Computing a non-scalar endomorphism, given a $\MaxOrder$ oracle} \label{alg:oneEndtoMaxOrder}
\begin{flushleft}
\hspace*{\algorithmicindent} \textbf{ Input: } A supersingular elliptic curve $E/\F_{p^2}$ and an access to an oracle for \MaxOrder.\\
\hspace*{\algorithmicindent} \textbf{ Output: } An endomorphism $\alpha \in \End(E)\setminus\Z$.
\end{flushleft}
\begin{algorithmic}[1]
\State Using the \MaxOrder\ oracle, obtain a maximal order $\mathcal{O} \simeq \End(E)$ in a quaternion algebra $B$;
\State $(\beta_i)_{i=1}^4 \gets$ a Minkowski-reduced basis of $\mathcal{O}$; \Comment{\cite{nguyen_low-dimensional_2009}}
\State $\theta \gets \beta_2$; \Comment{$\theta$ is a shortest non-scalar vector in $\mathcal{O}$}
\State $\ell \gets 1, N \gets 1$;
\While{$N < \Nrd(\theta)$} \label{alg:OneEndtoMaxOrder_step_while}
	\State $\ell \gets$ the next prime after $\ell$ which is coprime to $\Nrd(\theta)$;
	\State $N \gets \ell N$;
	\State $(P_\ell,Q_\ell) \gets $ a basis of the $\ell$-torsion $E[\ell]$; \Comment{\cite[Lemma 6.9]{bank_cycles_2019}}
	
	\State Using the oracle access, run Algorithm~\ref{alg_solveOneEndOrFingIso} on $E$, $\ell$, $B$, and $\mathcal{O}$ to obtain either
	\begin{itemize}
		\item a non trivial endomorphism $\alpha \in \End(E) \setminus \Z$,
		\item or an isomorphism $\iota_\ell: \mathcal{O}/\ell\mathcal{O} \stackrel{\sim}{\rightarrow} \End(E[\ell])$;
	\end{itemize}
	\If{ $\alpha \in \End(E) \setminus \Z $ was found } \Return $\alpha$; \EndIf
	
	\State $(P'_\ell,Q'_\ell) \gets (\iota_\ell(\theta)(P_\ell),\iota_\ell(\theta)(Q_\ell))$;
\EndWhile
\State $\alpha \gets \algoname{IsogenyInterpolation}(\Nrd(\theta),(P_\ell,Q_\ell)_\ell,(P'_\ell,Q'_\ell)_\ell)$; \Comment{\Cref{prop:hd-rep}}
\State \Return $\alpha$.
\end{algorithmic}
\end{algorithm}

\begin{proposition}[$\OneEnd$ reduces to \MaxOrder]\label{prop:OneEndtoMaxOrder}
\Cref{alg:oneEndtoMaxOrder} is correct and runs in probabilistic polynomial time in the length of the instance and in the length of the oracle's output.
\end{proposition}

\begin{proof}
By the references cited in the comments, each step is at most polynomial in $\log p$, in the length of the \MaxOrder\ oracle's output and in $\ell$ (whenever a prime $\ell$ is involved in the computation).
Therefore, to prove the claimed complexity, it remains only to establish bounds on the number of iterations of the loop at line \ref{alg:OneEndtoMaxOrder_step_while} and on the considered primes $\ell$.

By \cite{van_der_waerden_reduktionstheorie_1956}, a Minkowski-reduced basis of a lattice in dimension $4$ reaches all the successive minima.
Hence, $\Nrd(\theta)$ is the second minimum of $\mathcal{O}$, i.e. the first minimum of $\mathcal{O} \setminus \Z$.
Additionally, by Minkowski's second theorem, the product of the successive minima of $\mathcal{O}$ is smaller than $\gamma_4^2 \disc(\mathcal{O})$, where $\gamma_4$ is the Hermite constant in dimension $4$.
Thus we have that $\Nrd(\theta) \leq 2p^2$.
Therefore, by the prime number theorem, the while loop at line \ref{alg:OneEndtoMaxOrder_step_while} has $O(\log p)$ iterations and the largest $\ell$ considered is $O(\log p)$, proving the claimed complexity.

The correctness of \Cref{alg:oneEndtoMaxOrder} comes from the fact that, by \Cref{lemma:endo-image-on-torsion}, for any isomorphism $\iota: \mathcal{O} \stackrel{\sim}{\rightarrow} \End(E)$, the isomorphisms $\iota_\ell$ are its reductions modulo the corresponding $\ell$.
Hence, their evaluation provides the evaluation of $\iota(\theta)$ on the $N$-torsion subgroup $E[N]$ with $N > \Nrd(\theta)$.
\qed\end{proof}

\end{section}
\begin{section}{Reductions to the Endomorphism Ring problem}\label{sec:endring}

We now turn to proving that \Isogeny\ and \MaxEnd\ reduce in polynomial time to \EndRing, thereby completing the equivalence of all problems presented in Figure \ref{fig_generalcase}.
We shall proceed by proving the following sequence of reductions:

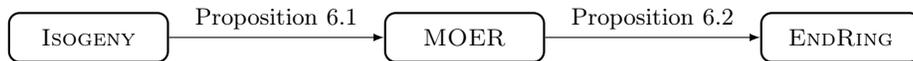
\begin{figure}[H]
\center
\begin{tikzpicture}
\node[probbox,minimum width=6.5em] (Isogeny) at (1.5,0) {\Isogeny};
\node[probbox,minimum width=6.5em] (MOER) at (6.5,0) {\MaxEnd};
\node[probbox,minimum width=6.5em] (EndRing) at (11.5,0) {\EndRing};

\draw [->,>=latex] (Isogeny) to node[midway,above] {\small{Proposition \ref{prop:Isogeny-to-MOER}}} (MOER);
\draw [->,>=latex] (MOER) to node[midway,above] {\small{Proposition \ref{prop:MOER-to-EndRing}}} (EndRing);
\end{tikzpicture}
\end{figure}

The main difficulty in reducing the \Isogeny\ problem between two curves to the \MaxOrder\ problem lies in translating a connecting ideal between two maximal orders, which are isomorphic to the endomorphism ring of the curves, into an isogeny between the curves.
Before the break of SIDH \cite{EC:CasDec23,EC:MMPPW23,EC:Robert23}, this process required first finding a more suitable ideal, using KLPT-type algorithms \cite{kohel_quaternion_2014}, which can be proven under GRH \cite{FOCS:Wesolowski21}, and then computing the corresponding isogeny using an elliptic curve with known endomorphism ring as a dictionnary between endomorphisms and quaternions.
Thanks to \Cref{prop:ideal-to-isogeny}, it is now possible to directly compute the isogeny corresponding to the connecting ideal.
However, one still needs to know an elliptic curve with an explicit basis of its endomorphism ring.
This is why, instead of reducing \Isogeny\ to \MaxOrder, we reduce it to the \MaxEnd\ problem, ensuring access to such a curve.

\begin{proposition}[$ \Isogeny \text{ reduces to } \MaxEnd$]
\label{prop:Isogeny-to-MOER}
Given access to a \MaxEnd\ oracle, one can solve the \Isogeny\ problem in time polynomial in the length of its input and in the length of the oracle's output.
\end{proposition}

\begin{proof}
Let $E_1$ and $E_2$ be two supersingular elliptic curves defined over $\mathbb{F}_{p^2}$.
For $i \in \{1,2\}$, a \MaxEnd\ oracle provides a maximal order $\mathcal{O}_i$ in a quaternion algebra $B_i \simeq B_{p,\infty}$ together with an isomorphism $\varepsilon_i: \mathcal{O}_i \stackrel{\simeq}{\longrightarrow} \End(E_i)$.
By proposition \ref{prop:quaternion-isomorphism}, one can compute in polynomial time an isomorphism $\varepsilon: B_2 \stackrel{\simeq}{\longrightarrow} B_1$.
Then $\mathcal{O}'_2 := \varepsilon(\mathcal{O}_2)$ is a maximal order in $B_1$ isomorphic to $\End(E_2)$.
Using \cite[Algorithm 3.5]{kirschmer_algorithmic_2010}, one can compute efficiently the connecting ideal $I = I(\mathcal{O}_1,\mathcal{O}'_2)$.
Finally, by Proposition~\ref{prop:ideal-to-isogeny}, one can compute the isogeny $\varphi_I: E_1 \to E_2$ in polynomial time.
\qed\end{proof}

We reduce \MaxEnd\ to \EndRing\ by adapting the strategy of \cite[Algorithm 6]{EC:EHLMP18}. The freedom to choose a model for $B_{p,\infty}$ in the definition of \MaxEnd\ allows us to eliminate all heuristics in the proof of \cite[Algorithm 6]{EC:EHLMP18}.
We recall that, using \Cref{prop:quaternion-isomorphism}, one can always translate a \MaxEnd\ solution into any target quaternion algebra where a maximal order is already known.

\begin{proposition}[$\MaxEnd \text{ reduces to } \EndRing$]
\label{prop:MOER-to-EndRing}
Given a supersingular elliptic curve $E$ defined over $\F_{p^2}$ together with a basis of its endomorphism ring $\End(E)$, one can solve the \MaxEnd\ instance corresponding to the curve $E$ in time polynomial in $\log p$ and in the length of the elements in the provided basis of $\End(E)$. 
\end{proposition}

\begin{proof}
Let $(\gamma_i)_{i=0}^4$ be a basis of the endomorphism ring of a supersingular elliptic curve $E$ defined over $\mathbb{F}_{p^2}$.
By \cite[Lemma 4 and Lemma 5]{EC:EHLMP18}, one can compute, in time polynomial in $\log p$ and in $\log \max_{i=1}^4(\deg(\gamma_i))$, a rational invertible linear transformation $F$ sending $(\gamma)_{i=1}^4$ to some orthogonal basis $(1,\alpha,\beta,\alpha \beta)$ of the endomorphism algebra $\End(E)\otimes\Q$.
In particular, $\alpha$ and $\beta$ satisfy $\alpha^2  < 0$, $\beta^2 < 0$ and $\alpha \beta = -\beta \alpha$.
Hence, it is isomorphic to the quaternion algebra $B = (\frac{\alpha^2,\beta^2}{\mathbb{Q}})$, with basis $(1,i,j,ij)$ such that $i^2 = \alpha^2, j^2 = \beta^2$ and $ij = -ji$.
Let $\varepsilon : \End(E) \otimes \Q \stackrel{\sim}{\to} B$ be the explicit isomorphism sending $(1,\alpha,\beta,\alpha\beta)$ to $(1,i,j,ij)$.
By applying $F^{-1}$ to $(1,i,j,ij)$ we get a maximal order $\mathcal{O} = \varepsilon((\gamma_i)_{i=1}^4)$ in $B$ isomorphic to $\End(E)$.
Finally, since $B$ is isomorphic to $\End(E)\otimes\Q$, which is itself isomorphic to $B_{p,\infty}$, the solution we found satisfies all the conditions of the \MaxEnd\ problem. 
\qed\end{proof}
\end{section}

\begin{section}{Worst-case to Average-Case reductions}
\label{sec:worst-case-to-average-case-reduction}

The goal of this section is to prove \Cref{theo:worst-case-to-average-case}: there are worst-case to average-case reductions between all of the listed fundamental problems of isogeny-based cryptography.
Recall that all of these problems take as input one or two elliptic curves, and the average case of these problems corresponds to random instances where the curves follow the stationary distribution (\Cref{def:stationary}).

To prove \Cref{theo:worst-case-to-average-case}, we follow the network of reductions summarized in \Cref{fig_wc2ac}. Once established, this network of reductions implies that $\OneEnd$ (in the worst-case) reduces to the average-case of any other problem. We then conclude from \Cref{theo:everything-is-equivalent}, which establishes that all problems reduce to \OneEnd.

Note that since the stationary distribution is computationally indistinguishable from the uniform distribution (\Cref{rem:indistinguishable}), the reductions also apply to the ``uniform'' version of the average-case problems.

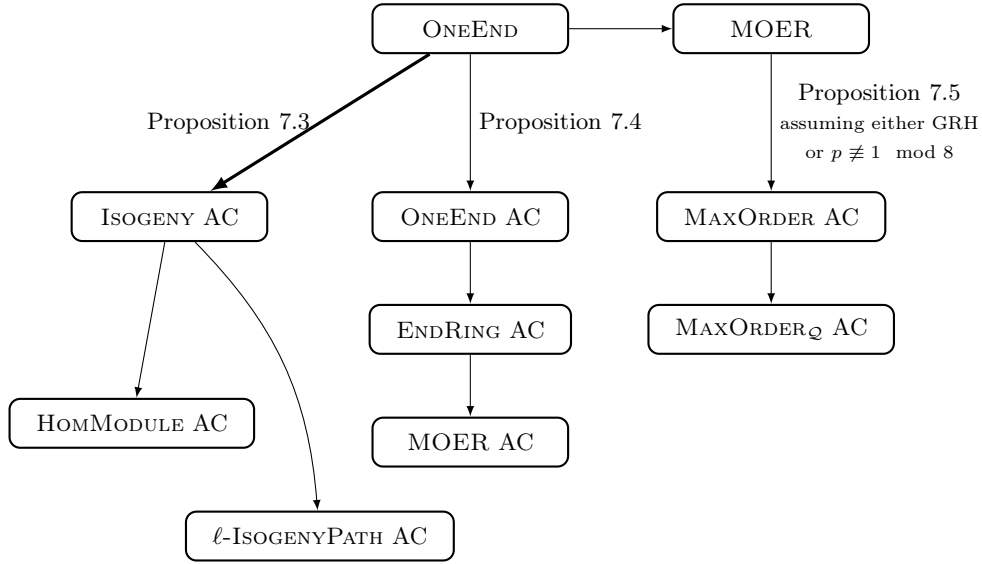
\begin{figure}[h]

\begin{tikzpicture}[box/.style={rectangle,draw},prb/.style={rectangle,draw,rounded corners}]

\node[probbox] (OneEnd) at (-1, 1.5) {\OneEnd};
\node[probbox] (OneEndAC) at (-0.5,-1) {\OneEnd\ AC};
\node[probbox] (lIsogenyAC) at (-2.0,-4.55) {\lIsogeny\ AC};
\node[probbox] (IsogenyAC) at (-4,-1) {\Isogeny\ AC};
\node[probbox] (EndRingAC) at (-0.5,-2.25) {\EndRing\ AC};
\node[probbox] (HomModuleAC) at (-4.5,-3.25) {\HomModule\ AC};
\node[probbox] (MaxOrderAC) at (3,-1) {\MaxOrder\ AC};
\node[probbox] (MaxOrderQAC) at (3,-2.25) {$\MaxOrder_\mathcal{Q}$ AC};
\node[probbox] (MOER) at (3,1.5) {\MaxEnd};
\node[probbox] (MOERAC) at (-0.5,-3.5) {\MaxEnd\ AC};

\draw[->,>=latex, very thick] (OneEnd) -- (IsogenyAC) node[midway,left] {Proposition~\ref{prop:worst-case-OneEnd-to-average-case-Isogeny}};
\draw[->,>=latex] (IsogenyAC) to[bend left = 20] (lIsogenyAC);
\draw[->,>=latex] (OneEnd) -- (OneEndAC) node[midway,right] {Proposition~\ref{prop:oneendWCtoOneEndAC}};
\draw[->,>=latex] (OneEndAC) -- (EndRingAC) node[midway,sloped,above] {};
\draw[->,>=latex] (IsogenyAC) -- (HomModuleAC) node[midway,sloped,above] {};
\draw[->,>=latex] (EndRingAC) -- (MOERAC) node[midway,sloped,above] {};

\draw[->,>=latex] (MaxOrderAC) -- (MaxOrderQAC);

\draw[->,>=latex] (OneEnd) -- (MOER) node[midway,sloped,above] {};
\draw[->,>=latex] (MOER) -- (MaxOrderAC) node[midway,right,align=center] {Proposition~\ref{prop:MOERWCtoMaxOrderAC} \\ \scriptsize{assuming either GRH} \\\scriptsize{or $p \not\equiv 1 \mod 8$} };

\end{tikzpicture}
\caption{Summary of worst-case to average-case reductions. 
Each arrow represents a probabilistic polynomial time reduction. Let ``AC'' label means ``average-case''.
All reductions are one-to-one except for the thick arrow, which requires on average fewer than $3$ oracle calls.
Each arrow is proved in the associated reference.
Arrows without reference are trivial reductions.}
\label{fig_wc2ac}
\end{figure}

\subsubsection{The straightforward reductions.}
In Figure \ref{fig_wc2ac}, every non-labeled arrow denotes a ``trivial'' reduction.
To be more precise, these reductions simply forward an instance for a given problem to an instance for a (at least as hard) variant of this problem.
In particular, the input of the average-case problems involved in those reductions always follows the same distribution.
For example, the input distribution of the average-case \Isogeny\ problem is a pair of two elliptic curves following the stationary distribution which is also the input distribution of the average-case \HomModule.
In addition, the solution we get for the harder problem directly includes a solution to the weaker one.
For instance, solving the \MaxEnd\ problem also yields a solution to the corresponding \EndRing\ instance. 
Therefore, all these reductions are trivial, and we only need to prove the reductions from worst-case \OneEnd\ to the average-case problems to complete the figure.
We shall address each proof of these non-trivial reductions in a dedicated subsection.

\begin{subsection}{\OneEnd\ reduces to average-case \Isogeny}

As there exist solutions of $\Isogeny$ of arbitrarily large degree, we ease the analysis of the reduction by making a bound explicit, as in \cite{EC:PagWes24}.

\begin{definition}[$\Isogeny_\lambda$]
Let $\lambda: \Z_{> 0} \rightarrow \Z_{>0}$ be a function.
The $\Isogeny_\lambda$ problem is a variant of the \Isogeny\ problem where the solution $\varphi$ needs to verify that $\log(\deg(\varphi)) \leq \lambda(\log p)$.
\end{definition}

In \cite{EC:PagWes24}, the authors have proven that one can solve the \OneEnd\ problem in expected polynomial time in $\log(p)$ and $\lambda(\log p)$ by calling on average at most $3$ times an $\Isogeny_\lambda$ oracle.
In this section, we adapt this reduction \cite[Algorithm 4]{EC:PagWes24} to ensure that it produces a \emph{semi average-case} \Isogeny\ instance,
in the sense that \emph{at least one} of the elliptic curves involved follows a distribution indistinguishable from the stationary distribution.
We then prove that the \emph{semi} average-case $\Isogeny_\lambda$ reduces to the average-case $\Isogeny_\lambda$, and deduce the claimed expected polynomial time reduction from \OneEnd\  to $\Isogeny_\lambda$.

\begin{proposition}
\label{prop:OneEndWC-to-almost-IsogenyAC}
Let $c_1,c_2 > 0$, and consider the following variant of \cite[Algorithm 4]{EC:PagWes24} where
\begin{itemize}
	\item the parameter $\varepsilon$ is smaller than $1/p$,
	\item the length of the non-backtracking random walks in the $3$-isogeny graph is $n$, where $n$ satisfies $n \geq c_1\log(p) - c_2\log(\varepsilon)$.
\end{itemize}
There exist absolute computable constants $c_1$ and $c_2$ such that this algorithm computes an endomorphism in expected polynomial time in $\log p$, $\lambda(\log p)$ and $n$ with at most $3$ calls to an $\Isogeny_\lambda$ oracle.
In addition, these calls are done on \emph{semi} average-case instances.
\end{proposition}

\begin{proof}
For $p > 6$, the proof of \cite[Theorem 8.6]{EC:PagWes24} still applies to this variant of \cite[Algorithm 4]{EC:PagWes24} $\varepsilon < \frac{1}{6}$, as we can set $c_1$ and $c_2$ such that $n$ is larger than $\lceil 2 \log_3(p) - 4\log_3(\varepsilon) \rceil$, as needed. In particular, the algorithm computes an endomorphism in expected polynomial time in $\log p$, $\lambda(\log p)$ and $n$ with at most $3$ calls to an $\Isogeny_\lambda$ oracle.

We now prove the statement about the distribution of the instances given to the oracle.
Let $E$ be the $\OneEnd$ instance (a supersingular elliptic curve over $\mathbb{F}_{p^2}$).
Let us denote by $\mathcal{O}_\Isogeny$ the $\Isogeny_\lambda$ oracle we have access to.
Following \cite[Algorithm 4]{EC:PagWes24}, a point $S \in E[2]$ is fixed, and each call to the oracle $\mathcal{O}_\Isogeny$ is done on an instance $(E'',E)$ where $E''$ is the codomain of the composition of isogenies $\nu \circ \varphi$ where $\varphi: E \rightarrow E'$ is a non-backtracking random walk in the $3$-isogeny graph of length $n$ and $\nu:E' \rightarrow E''$ is an isogeny of kernel $\langle \varphi(S) \rangle$. As $\nu$ and $\varphi$ have coprime degree, the distribution of $E''$ is the same as the codomain of $\varphi' \circ \nu'$ where $\nu':E \rightarrow E/\langle S \rangle$, and $\varphi' : E/\langle S \rangle \to E''$ is a non-backtracking random walk in the $3$-isogeny graph of length $n$.
By \cite[Theorem 11]{EC:BCCDFL23}, we can set $c_1$ and $c_2$ to ensure that for any $n \geq c_1\log(p) - c_2\log(\varepsilon)$, the distribution of $E''$ is an statistical distance at most $\varepsilon$ from the stationary distribution.
In particular, each call to the oracle $\mathcal{O}_\Isogeny$ is done on \emph{semi} average-case instances.

\qed\end{proof}

\begin{proposition}[$\OneEnd \text{ reduces to average-case } \Isogeny$]
\label{prop:worst-case-OneEnd-to-average-case-Isogeny}
Solving an instance of the worst-case $\OneEnd$ problem can be reduced in expected polynomial time in $\log p$ and $\lambda(\log p)$ to solving $3$ average instances of the $\Isogeny_\lambda$ problem.
\end{proposition}

\begin{proof}
First, we show that a \emph{semi} average instance of $\Isogeny_{\lambda_{c}}$ reduces to an average instance of $\Isogeny_{\lambda}$, where $\lambda_{c}(n) := \lambda(n) + 2cn + 1$, with $c$ the $O$-constant of \Cref{cor_randomwalk}.
Note that $c$ depends only on $p$ and on some positive parameter $\varepsilon$, which we fix to be $\varepsilon := 1/p$.
Let $E_0$ and $E_1$ be two supersingular elliptic curves defined over $\mathbb{F}_{p^2}$ such that $E_1$ is sampled from the stationary distribution.
By \Cref{cor_randomwalk}, one can compute a random walk  $\eta: E_0 \rightarrow E_2$ in the $2$-isogeny graph of length $\lceil \tau(p,\varepsilon) \rceil$ such that the distribution followed by $E_2$ is at total variation distance at most $\varepsilon$ to the stationary distribution.
Then, as $\varepsilon = 1/p$, the two distributions are computationally indistiguishable.

Moreover, since $\tau(p,\varepsilon) = O(\log(p) - \log(\varepsilon))$ and $\varepsilon = 1/p$, we have that $\lceil \tau(p,\varepsilon)\rceil \leq 2c\log(p) + 1$.
Then, the random $2^{\lceil \tau(p,\varepsilon) \rceil}$-walk isogeny $\eta$ verifies that $\log(\deg(\eta)) \leq 2c \log (p) +1$.

A call to an $\Isogeny_\lambda$ oracle on the pair $(E_2,E_1)$, which is indistinguishable from an average-case instance, returns an isogeny $\varphi: E_2 \rightarrow E_1$ such that $\log (\deg(\varphi)) \leq \lambda(\log (p))$.
Then, the isogeny $\psi := \varphi \circ \eta: E_0 \rightarrow E_1$ verifies that $\log(\deg (\psi)) \leq 2c\log(p) + \lambda(\log (p)) + 1 = \lambda_c(\log(p))$.
Thus the isogeny $\psi$ is a solution to the initial $\Isogeny_{\lambda_{c}}$ problem corresponding to $(E_0,E_1)$.
This concludes the proof that a \emph{semi} average instance of $\Isogeny_{\lambda_{c}}$ reduces to an average instance of $\Isogeny_{\lambda}$.

We can now conclude the proof: by Proposition \ref{prop:OneEndWC-to-almost-IsogenyAC} and because $\lambda_{c}(\log (p))$ is polynomial in $\lambda(\log p)$ and $\log (p)$, the \OneEnd\ worst-case problem reduces in expected polynomial time in $\log(p)$ and $\lambda(\log(p))$ to $3$ \emph{semi} average instances of $\Isogeny_{\lambda_{c}}$, which themselves reduce to $3$ average instances of $\Isogeny_\lambda$ in polynomial time as proven above.
\qed\end{proof}

\end{subsection}

\begin{subsection}{\OneEnd\ reduces to average-case \OneEnd}

The reduction presented in this subsection is analogous to the most folkoric methods for self-reducing the \Isogeny\ problem from the worst-case to the average-case, leveraging the rapid mixing properties of isogeny graphs.

\begin{proposition}[$\OneEnd \text{ reduces to average-case } \OneEnd$]\label{prop:oneendWCtoOneEndAC}
Solving an instance of the worst-case \OneEnd\ problem can be reduced to solving an average-case instance of the \OneEnd\ problem in time polynomial in $\log p$ and in the length of the averace-case solution.
\end{proposition}

\begin{proof}
Let $E$ be a supersingular elliptic curve defined over $\mathbb{F}_{p^2}$.
Let $\eta: E \rightarrow E'$ be a random-walk in the $2$-isogeny graph of length $n = \lceil \tau(p,1/p) \rceil$, so that the distribution followed by $E'$ is asymptotically indistinguishable from the stationary distribution by \Cref{cor_randomwalk}.
Then, from a solution $\theta: E' \to E'$ to the average-case \OneEnd\ instance corresponding to the curve $E'$, one obtains a non trivial endomorphism $\hat \eta \circ \theta \circ \eta : E \to E$  which is a solution to the worst-case instance of \OneEnd\ given by $E$.
Indeed, this endomorphism is non trivial; otherwise there exists $n \in \mathbb{Z}$ such that $\hat \eta \circ \theta \circ \eta = n$ so $ [\deg \eta] \circ \theta = n$, thus $\theta$ is a scalar endomorphism which is a contradiction.
\qed\end{proof}

\end{subsection}

\begin{subsection}{\MaxEnd\ reduces to average-case \MaxOrder}

The main challenge in proving unconditional reductions to the \MaxOrder\ problem lies in leveraging the information obtained from the quaternion world to aid in isogeny computation, without having an access to a dictionary between endomorphisms and quaternions.
Indeed, we recall that when $p \equiv 1 \mod 8$, there is currently no known polynomial time algorithm free from GRH that can compute a supersingular elliptic curve defined over $\bar \F_p$ together with an embedding of its endomorphism ring into some quaternion algebra isomorphic to $B_{p,\infty}$.
In Section~\ref{sec:oneend to maxorder}, we address this difficulty by ``locally'' computing this embedding for sufficiently many primes, allowing us to apply the recent \algoname{IsogenyInterpolation} algorithm.
Unfortunately, this method requires solving the \MaxOrder\ problem for elliptic curves which are close to each other in the same isogeny graph.
Thus, it cannot be turned into a reduction to the average-case \MaxOrder\ problem.
For this reason, the reduction presented below requires the construction of a curve $E_0$ for which a solution of \MaxEnd\ is known. This requires either $p \not\equiv 1 \mod 8$, or to assume GRH.

\begin{proposition}[$\MaxEnd \text{ reduces to average-case } \MaxOrder$]\label{prop:MOERWCtoMaxOrderAC}
An instance of the worst-case \MaxEnd\ can be reduced to an average instance of the \MaxOrder\ problem in polynomial time in the length of the input.
If $p \equiv 1 \mod 8$, this result assumes GRH.
\end{proposition}

\begin{proof}
By \cite[Proposition 3]{EC:EHLMP18}, one can compute in polynomial-time a curve $E_0$ together with a quaternionic order $\mathcal{O}_0$ and an isomorphism $\varepsilon_0: \mathcal{O}_0 \stackrel{\sim}{\to} \End(E_0)$ (i.e., a solution to $\MaxEnd$). This result assumes GRH in the case where $p \equiv 1 \mod 8$.
We denote by $B$ the quaternion algebra containing $\mathcal{O}_0$.

Let $E$ be a supersingular elliptic curve defined over $\mathbb{F}_{p^2}$.
Let us solve the \MaxEnd\ problem for the elliptic curve $E$ calling once a \MaxOrder\ oracle on an average elliptic curve $E'$.

Let $N= \prod_{i=1}^n\ell_i$, where $\ell_i$ is the $i$-th smallest prime number, and let $\eta: E \to E'$ be a random $N$-walk.
We define $N$ and $\eta$ this way to then use \cite[Lemma 7.1]{FOCS:Wesolowski21} efficiently on $\hat \eta$.
By \Cref{cor_randomwalk}, by choosing $n$ such that $\log(N) \geq \tau(p,1/p)$, one can ensure that $E'$ follows a distribution statistically indistinguishable from the stationary distribution.
In particular, as $\tau(p,1/p) = O(\log p)$, by the prime number theorem, it is sufficient to consider primes up to some $\ell_n = O(\log(p))$.
Thus, the computation of $\eta$ takes a time polynomial in $\log p$.

Let us now solve the worst-case \MaxEnd\ instance corresponding to $E$ from a solution $\mathcal{O}'$ to the average instance given by $E'$.
Thanks to Proposition~\ref{prop:quaternion-isomorphism}, one can assume that $\mathcal{O}'$ is a maximal order in the quaternion algebra $B$.

First, we compute a connecting ideal $I$ between $\mathcal{O}'$ and $\mathcal{O}_0$, \cite[Algorithm 3.5]{kirschmer_algorithmic_2010}, and the corresponding isogeny using Proposition~\ref{prop:ideal-to-isogeny}.
By running \cite[Algorithm 8]{EC:DLRW24}, see also \cite[Algorithm 4]{EC:EHLMP18}, where the final division is done using Proposition~\ref{prop:isogeny-division}, one obtains an isomorphism between $\End(E')$ and a maximal order $\mathcal{O}'$ in polynomial time.
In particular, the endomorphisms in the basis are efficiently represented.
Then by using \cite[Lemma 7.1]{FOCS:Wesolowski21} on the isogeny $\hat \eta$, one can compute, in polynomial time in $\log p$, the corresponding left $\mathcal{O}'$ ideal $I_{\hat \eta}$ such that $\mathcal{O}_R(I) = \mathcal{O}$.
Hence, using \cite[Algorithm 8]{EC:DLRW24} again, we obtain an explicit isomorphism between $\End(E)$ and a maximal order in $B$.
\qed\end{proof}

\end{subsection}

\begin{subsection}{Proof of Theorem \ref{theo:worst-case-to-average-case}}

We can now turn to the proof of the main theorem of this section.

\begin{proof}[of Theorem \ref{theo:worst-case-to-average-case}]
Let $P$ and $Q$ be two problems chosen from the problems \lIsogeny, \Isogeny, \EndRing, \OneEnd, \MaxEnd, \MaxOrder, $\MaxOrder_\mathcal{Q}$, and \HomModule.

By Theorem~\ref{theo:everything-is-equivalent}, if $P$ is not \lIsogeny, we have a probabilistic polynomial time reduction from $P$ in the worst-case to \OneEnd\ in the worst-case.
Otherwise, assuming the generalised Riemann hypothesis, there is a probabilistic polynomial time reduction from \lIsogeny\ in the worst-case to \OneEnd\ in the worst-case by \cite{FOCS:Wesolowski21} and \cite{EC:PagWes24}.
Then using the results summarized in \Cref{fig_wc2ac}, there is a probabilistic polynomial time reduction from \OneEnd\ in the worst-case to $Q$ in the average-case.
\qed\end{proof}

\end{subsection}

\end{section}

\bibliographystyle{alpha}
\bibliography{refs,abbrev0,crypto}

\end{document}